\documentclass{amsart}
\usepackage[foot]{amsaddr}
\usepackage{amsthm}

\theoremstyle{plain}
\newtheorem{theorem}{Theorem}[section]
\newtheorem{lemma}[theorem]{Lemma}
\newtheorem{corollary}[theorem]{Corollary}

\theoremstyle{definition}

\usepackage{thm-restate}
\usepackage{nicefrac}
\usepackage[ruled,vlined,norelsize]{algorithm2e}
\usepackage{tikz,pgfplots}
\pgfplotsset{compat=1.11}
\usetikzlibrary{matrix,calc}
\newcommand{\Exx}[2]{\mbox{\rm\bf E}_{#1}\bigl[#2\bigr]}
\newcommand{\OPT}{\textsc{Opt}}
\newcommand{\ALG}{\textsc{Alg}}
\newcommand{\lfrac}[2]{#1/#2}
\newcommand{\Ex}[1]{\mbox{\rm\bf E}\left[#1\right]}
\newcommand{\Exp}[1]{\mbox{\rm\bf E}\left[#1\right]}
\renewcommand{\P}{\mbox{\rm\bf Pr}}

\newcommand{\Prob}[1]{\mbox{\rm\bf Pr}\left[#1\right]}
\newcommand{\intd}{\normalfont{\mbox{\,d}}}
\newcommand{\U}{U}
\newcommand{\E}{\mbox{\rm\bf E}}
\newcommand{\Hn}{\mathcal{H}}
\newcommand{\e}{e}
\newcommand{\ct}{\ensuremath{\tau}}
\newcommand{\set}{\leftarrow}

\newcommand{\notnicefrac}[2]{#1/#2}
\usepackage{todonotes}
\usepackage[font=scriptsize]{subcaption}

\usepackage{psfrag}

\usepackage{natbib}
 \bibpunct[, ]{[}{]}{,}{n}{}{,}%
 
\newcommand{\ratiouniform}{2.965}
\newcommand{\ratioarbitrary}{6.052}
\newcommand{\ratiouniformbad}{8.122}
\newcommand{\ratiolower}{2.148}
\newcommand{\ratiounknown}{48}

\newcommand{\Halmos}{}

\begin{document}
 
\title[Hiring Secretaries over Time]{Hiring Secretaries over Time:\\ The Benefit of Concurrent Employment}

\author{Yann Disser}
\address[Y.~Disser]{TU Darmstadt, Graduate School CE, Germany}
\email[Y.~Disser]{disser@mathematik.tu-darmstadt.de}
\thanks{First author supported by the `Excellence Initiative' of the German Federal and State Governments and the Graduate School~CE at TU~Darmstadt.}

\author{John Fearnley}
\author{Martin Gairing}
\address[J.~Fearnley and M.~Gairing]{University of Liverpool, U.K.}
\email{john.fearnley@liverpool.ac.uk}\email{gairing@liverpool.ac.uk}

\author{Oliver G\"obel}
\address[O.~G\"obel]{RWTH Aachen University, Germany}
\email{goebel@cs.rwth-aachen.de}

\author{Max Klimm}
\address[M.~Klimm]{Humboldt University Berlin, Germany}
\email{max.klimm@hu-berlin.de}
\thanks{Research of the fourth author was carried out in the framework of {\sc Matheon} supported by Einstein Foundation Berlin.}

\author{Daniel Schmand}
\address[D.~Schmand]{RWTH Aachen University, Germany}
\email{daniel.schmand@oms.rwth-aachen.de}

\author{Alexander Skopalik}
\address[A.~Skopalik]{University of Paderborn, Germany}
\email{skopalik@mail.uni-paderborn.de}

\author{Andreas T\"onnis}
\address[A.~T\"onnis]{University of Bonn, Germany}
\email{atoennis@uni-bonn.de}

\begin{abstract}

We consider a stochastic online problem where $n$ applicants arrive over time, one per time step. 
Upon arrival of each applicant their cost per time step is revealed, and we have to fix the duration of employment, starting immediately. This decision is irrevocable, i.e., we can neither extend a contract nor dismiss a candidate once hired. In every time step, at least one candidate needs to be under contract, and our goal is to minimize the total hiring cost, which is the sum of the applicants' costs multiplied with their respective employment durations.
We provide a competitive online algorithm for the case that the applicants' costs are drawn independently from a known distribution. Specifically, the algorithm achieves a competitive ratio of $\ratiouniform$ for the case of uniform distributions. For this case, we give an analytical lower bound of~$2$ and a computational lower bound of $\ratiolower$.
We then adapt our algorithm to stay competitive even in settings with one or more of the following restrictions: $(i)$ at most two applicants can be hired concurrently; $(ii)$~the distribution of the applicants' costs is unknown; $(iii)$ the total number $n$ of time steps is unknown.
On the other hand, we show that concurrent employment is a necessary feature of competitive algorithms by proving that no algorithm has a competitive ratio better than $\Omega(\sqrt{n} / \log n)$ if concurrent employment is forbidden.

\end{abstract}

\keywords{
online algorithm, stopping problem, prophet inequality, Markov chain, secretary problem
}

\subjclass[2000]{Primary: 60G40, 62L15, 68W27; secondary: 68W40, 68Q87}


\maketitle

\markleft{DISSER ET AL.}

\newpage
\section{Introduction}
The theory of optimal stopping is concerned with problems of finding the best points in time to take a certain action based on a sequence of sequentially observed random variables. Problems of this kind are ubiquitous in the area of operations research, e.g., when hiring, selling, purchasing, or procurement decisions are made based on the partial observation of a sequence of offers with known statistical properties. In one of the most basic stopping problems, a gambler sequentially observes realizations $x_1 \sim X_1$, $x_2 \sim X_2$, \dots of a series of independent random variables. After being presented a realization $x_i \sim X_i$, the gambler has to decide immediately whether to keep the realization $x_i$ as a prize, or to continue gambling hoping for a better realization. For this setting, the famous prophet inequality due to Krengel, Sucheston, and Garling (cf.~\cite{krengel77,krengel78}) asserts that the best stopping rule of the gambler achieves in expectation at least half the optimal outcome of a prophet that foresees the realizations of all random variables and, thus, gains the expected maximal realization of all variables.

After the surprising result of Krengel et al., prophet-type inequalities were provided for several generalizations of their model, including settings where both the gambler and the prophet may stop multiple times (cf.~Kennedy~\cite{kennedy1987}, Alaei~\cite{alaei2014}), settings where both choose a set subject to matroid constraint (cf.~Kleinberg and Weinberg~\cite{kleinberg2012}), polymatroid constraints (cf.~D\"utting and Kleinberg~\cite{duetting2015}), and general constraints (cf.~Rubinstein~\cite{rubinstein2016}).

In light of this remarkable progress in establishing prophet-type inequalities for various stochastic environments, two remarks are in order. First, the known results consider maximization problems only. While obviously important as a model for situations where, e.g., items are to be sold and offers for the items arrive over time, they do not capture the ``dual'' problem where items need to be procured. In fact, minimization problems in stochastic environments turn out to be much harder, as any stopping rule does not allow for a constant factor approximation compared to the prophet's outcome, even in the most basic case of single stopping and i.i.d.\ distributions (cf.~Esfandiari et al.~\cite{esfandiari2015}). Second, the models above are inherently \emph{static} in the sense that the objective depends only on the set of chosen realizations \emph{at the end} of the sequence. This is a reasonable assumption when the underlying selling or purchase decisions have a long-term impact, and the time during which the sequence of random variables is observed can be neglected. On the other hand, they fail to capture the natural situation where realizations are observed for a long period of time and selling or procurement decisions are taking effect even while further offers are observed.

To illustrate the key differences between static and dynamic settings, consider a firm that in each time step needs to be able to perform a certain task in order to be operational. Traditionally, the firm could advertise a position and hire an applicant able to perform the task. Assuming that the firm strives to minimize labour cost, this leads to a (static) prophet-type problem where the costs of the applicants are drawn from a distribution and the firm strives to minimize the realized costs. Alternatively, online marketplaces like oDesk.com and Freelance.com provide the opportunity to hire applicants with a limited contract duration and to possibly hire another contractor when a new offer with lower cost arrives. The constant rise of the revenue generated by these platforms (reaching 1 billion USD in 2014) suggests that the latter approach has growing economic importance \cite{verroios2015}.

Hiring employees for a limited amount of time leads to a new kind of stopping problem where 
the ongoing observation period overlaps with the duration of contracts, and active contracts need to be maintained over time while receiving new offers.
To model these situations, we study a natural setting where at least one contract needs to be active at each point in time, while there is no additional benefit of having more than one active contract.\footnote{We discuss a relaxation of the strict covering constraint in Section~\ref{sec:conclusion}.}
This covering constraint renders it beneficial to accept good offers even when other contracts are still active, and a key challenge is to manage the tradeoff between accepting good offers while avoiding contract overlaps.



Specifically, we assume that in every time step $i \in [n]$ we observe the cost of the $i$-th applicant $x_i$, where the values $x_i$ are drawn i.i.d.\ from a common distribution $X$. 
%
In each time step~$i$, we have to decide on a number of time steps~$t_i$ for which to hire the $i$-th applicant. 
This duration is fixed irrevocably at time $i$ and extension or shortening of this duration is impossible later on. 
Hiring applicant $i$ with realized cost $x_i$ results in costs of $x_i t_i$. 
We are interested in minimizing the expected total hiring cost $\Exx{x_1 \sim X,\dots,x_n \sim X}{\sum_{i=1}^{n} {t_i x_i}}$, subject to the constraint that at least one applicant is under employment at all times.



\subsection{Results and Outline}

When the total number of time steps and the distribution are known, the dynamic stopping problem considered in this paper can be solved by a straightforward dynamic program (DP). The DP maintains a table of $n^2$ optimal threshold values depending on the number of remaining covered and uncovered time steps. Like other optimal solutions for similar stochastic optimization problems, the DP suffers from the fact that it relies on the exact knowledge of the distribution and the number of time steps, and does not allow to quantify the optimal competitive ratio.

The results we give in this paper address these shortcomings. We
give online algorithms with constant competitive ratios, and in doing so, we
prove that the optimal online algorithm also gives a constant competitive ratio
for any cost distribution that is known upfront. Our techniques are robust with respect to incomplete information and can be extended to the case where the cost
distribution and/or the total number of time steps is unknown, while still providing a constant competitive ratio.
Furthermore, our approach is conceptually simple, efficient and not tailored to specific distributions.

For ease of exposition, we present
our algorithm in incremental fashion starting with a simplified version for
uniform distributions in \S~\ref{sec:uniform_first}. The algorithm maintains
different threshold values over time and hires applicants when their realized
cost is below the threshold. By relating the execution of the algorithm with a
Markov chain and by analyzing its hitting time, we bound the competitive ratio
of the algorithm. In \S~\ref{sec:uniform2}, we refine the algorithm and its
analysis to show that it is $\ratiouniform$-competitive in the uniform case. We provide
an analytical lower bound of $2$ for the best possible competitive ratio via a relaxation to the Cayley-Moser-Problem (cf.~\citet{moser1956}), and we give a computational lower bound of $2.14$. 

Subsequently, in \S~\ref{arbitrary_distributions}, we generalize the algorithm to arbitrary distributions. Here, the main technical difficulty is to obtain a good estimation of the offline optimum. As we bound the offline optimum by a sum of conditional expectations given that the value lies in an intervals bounded by exponentially decreasing quantiles of the distribution, we are able to derive a competitive ratio of $\ratioarbitrary$.

In \S~\ref{sec:unknown_distribution}, we further generalize our techniques to
give a constant competitive algorithm for the case where the distribution is
unknown a priori. The main idea of the algorithm is to approximate the quantiles
of the distribution by sampling. 

Finally, in \S~\ref{sec:sequential_employment}, we show that our algorithms
remain competitive in the case that at most two applicants may be employed
concurrently. We also extend our results to the case where the
total number of applicants is unknown. 
In contrast to this, we show that the best possible online algorithm without concurrent employment has competitive ratio $\Theta(\notnicefrac{\sqrt{n}}{\log n})$, even for uniform distributions.

To improve readability, we relegate the formal analysis of the underlying Markov chains to \S~\ref{sec:markov}.

\subsection{Related Work}
The interest in optimal stopping rules for sequentially observed random experiments dates at least as far as to \citet{cayley1875} who asked for the optimal stopping rule when $n$ tickets are drawn without replacement from a known pool of $N$ tickets with different rewards. See also \citet{ferguson1989} for more historical notes on this problem. Cayley solved this problem by backwards induction, an approach later formalized by Bellman~\cite{bellman1954}. \citet{moser1956} studied Cayley's problem for the case that $N$ is large and the $N$ rewards are equal to the first $N$ natural numbers. In that case, the problem can be approximated by $n$ draws from the uniform distribution and Moser provided an approximation of the corresponding threshold values of the optimal stopping rule. For similar results for other distributions, see \citet{GilbertMosteller/66,guttman1960} and \citet{karlin1962}. In \S~\ref{sec:lower_bound}, we will use the asymptotic behavior of the threshold  due to \citet{GilbertMosteller/66} to obtain a lower bound for our problem.

Krengel, Sucheston, and Garling (cf.~\cite{krengel77,krengel78}) studied optimal stopping rules for arbitrary independent, non-negative, but not necessarily identical random variables. Their famous prophet inequality asserts that the expected reward of a gambler who follows the optimal stopping rule (that can still be found using backwards induction) is at least half the expected reward of a prophet who knows all realizations beforehand and will stop the sequence at the highest realization. 
Samuel-Cahn \cite{samuel-cahn1984} showed that the same guarantee can be obtained by a simple stopping rule that uses a single threshold rather than $n$ different thresholds as the solution of the dynamic program. Hill and Kertz~\cite{hill1992} surveyed some variations of the problem.

More recently, Alaei~\cite{alaei2014} considered the setting where both the prophet and the gambler stop $k \in  \mathbb{N}$ times and receive the sum of their realizations as rewards and gave an algorithm with competitive ratio $1 - \frac{1}{\sqrt{k+3}}$. For a more general setting in which the selection of both the gambler and the prophet is restricted by a matroid constraint, Kleinberg and Weinberg \cite{kleinberg2012} showed a tight competitive ratio of $1/2$. D\"utting and Kleinberg \cite{duetting2015} generalized this result further to polymatroid constraints. \citet{goebel2014} studied a prophet inequality setting where a solution is feasible if it forms an independent set in an underlying network. They gave an online algorithm that achieves a $\mathcal{O}(\rho^2 \log n)$-approximation where $\rho$ is a structural parameter of the network. Very recently, \citet{rubinstein2016} studied the problem for general downward-closed constraints. He gave a $\mathcal{O}(\log n \log r)$-approximation where $r$ is the cardinality of the largest feasible set and showed that no online algorithm can be better than a $\mathcal{O}(\log n / \log \log n)$-approximation. For a generalization towards non-linear valuations functions, see \citet{rubinstein2017}.

The recent interest in prophet inequalities is due to an interesting connection to mechanism design problems that was first made by Hajiaghayi et al.~\cite{hajiaghayi2007}. They remarked that threshold rules used to prove prophet inequalities correspond to thruthful online mechanisms with the same approximation guarantee as the prophet inequality. Chawla et al.~\cite{chawla2010} noted that posted pricing mechanisms for revenue maximization can be derived from prophet inequalities by using the framework of virtual values due to \citet{myerson1981}. As our algorithms operate on the basis of threshold values as well they can also be turned into truthful mechanisms. However, the exact properties of these mechanisms deserve further investigation.

\citet{esfandiari2015} considered the minimization version of the classical prophet inequality setting. They showed that even for i.i.d.\ random variables, no stopping rule can achieve a constant approximation to the cost of a prophet. This is in contrast to our results for the dynamic prophet inequality setting as we obtain a constant factor approximation even without knowledge of the distributions or $n$.

Further related are secretary problems (cf.~Ferguson~\cite{ferguson1989} for a review), and in particular secretary problems where the values are drawn from i.i.d.\ distributions as considered by Bearden~\cite{bearden}. The main difference to our model is that in secretary problems the objective is to maximize the probability of selecting the best outcome. Yet, our algorithm developed in \S~\ref{sec:unknown_distribution} for solving the case of unknown distributions is reminiscent of the optimal stopping rules for secretary problems as it also employs a sampling phase in which the distribution is learned before hiring an applicant. 
Very recently, Fiat et al.~\cite{fiat2015} studied a dynamic secretary problem where secretaries are hired over time. In contrast to our work, they consider a maximization problem, and the contract duration is fixed.

\section{Preliminaries}
\label{sec:model}

For a natural number $n \in \mathbb{N}$ let $[n] = \{1,\dots,n\}$. We consider a
sequence $x_1 \sim X$, $x_2 \sim X$, $\dots$, $x_n \sim X$ of~$n$ i.i.d.\ random variables drawn from a
probability distribution~$X$.
Throughout this work, we assume that $X$ is a continuous distribution with
cumulative distribution~$F$ and probability density function $f$. Moreover, we
assume that $X$ assigns positive probability to non-negative values only, i.e.,
$F(0) = 0$.
In every time step $i \in [n]$ the cost $x_i$ of the $i$-th
applicant is revealed and we must decide the number of time steps $t_i$ the applicant is hired. 
The duration of the employment $t_i$ is fixed irrevocably at time~$i$; no extension or shortening of this duration at any further point in time is possible. 
Hiring applicant~$i$ with realized cost $x_i$ for~$t_i$ time steps results in costs of $x_i t_i$. 
The objective is to minimize the expected total cost of hired applicants $\E[{\sum_{i \in [n]} t_i x_i}] := \Exx{x_1 \sim X, \dots, x_n \sim X}{\sum_{i \in [n]} t_i x_i}$ subject to the constraint that at least one applicant is employed at each point in time $i \in [n]$, i.e., $\max_{j \leq i} \{j+t_j\} \geq i$ for all $i \in [n]$. 

This is an online problem since, at time~$i$, we only know about the realizations $x_1,\dots,x_i$ up to time $i$ and have to base our decision about the hiring duration~$t_i$ of the $i$-th applicant only on this information and previous hiring decisions $t_1,\dots,t_{i-1}$. We are interested in obtaining online algorithms that perform well compared to an omniscient prophet. Let $\OPT_n$ be the expected cost of an \emph{optimal offline algorithm} (i.e., a prophet) knowing the $n$ realizations in advance and let $\ALG_n$ be the expected cost of a solution of an online algorithm. 
Then the competitive ratio of the online algorithm $\ALG_n$ is defined as $\lim\sup_{n \in \mathbb{N}} \lfrac{\Ex{\ALG_n}}{\Ex{\OPT_n}}$. 
We call an algorithm competitive if its competitive ratio is constant, and call it \emph{strictly competitive} if even $\sup_{n \in \mathbb{N}} \lfrac{\Ex{\ALG_n}}{\Ex{\OPT_n}}$ is constant.

We use well-known facts from higher order statistics of random variables to obtain the following.


\begin{restatable}{proposition}{proopt}
\label{pro:opt}
The expected total cost of an optimal offline algorithm is $\Ex{\OPT_n} = \sum_{i \in [n]} \int_0^{\infty} \bigl(1 - F(x)\bigr)^i \intd x$.	
\end{restatable}

\begin{proof}{Proof.}
In every step, the optimal offline algorithm employs the applicant with the lowest cost that has arrived so far. 
We have
\begin{align*}
\Ex{\OPT_n}
&= \E \Bigl[\sum\nolimits_{i \in [n]} \min\nolimits_{j \in [i]}\{x_j\} \Bigr]\\[6pt]
&= \sum_{i \in [n]} \E \Bigl[\min\nolimits_{j \in [i]} \{x_j\}\Bigr]\\
&= \sum_{i \in [n]} \int_0^{\infty} \P \Bigl[\min\nolimits_{j \in [i]} \{x_j\} > x\Bigr] \,\intd x\\
 &=  \sum_{i \in [n]} \int_0^{\infty} \bigl(1-F(x)\bigr)^{\!i} \,\intd x,
\end{align*}
as claimed. \hfill \Halmos
\end{proof}

\section{An Optimal Online Algorithm}
\label{sec:dp}

We begin by describing an optimal online algorithm that uses dynamic
programming.
Let $C(i, j)$ denote the expected overall cost if there are $i$
time steps remaining, and if the next $j$ time steps are already covered by an existing
contract. As a boundary condition, we have that $C(i, i) = 0$ for all $i$, since
in this case no further applicants need to be hired.

Suppose that $C(i', j')$ has already been computed for all $i' < i$ and all
$j' \le i'$. First we describe how to compute $C(i, 0)$. 
Suppose that we draw an 
applicant with cost $x$. Since there are no existing contracts, we
must hire this applicant for at least one time step, and we will obviously hire this
applicant for at most $i$ time steps. If we hire the applicant for $r$ time steps, our
overall cost will be $r x + C(i-1, r-1)$. Thus, the optimal cost for an
applicant costing $x$ can be written as
\begin{align*}
\min_{1 \le r \le i}\left\{ rx + C(i-1, r-1) \right\}.
\end{align*}
Therefore, we have
\begin{align}
\label{eqn:dynamic0}
C(i, 0) = \int_{0}^{\infty} \min_{1 \le r \le i}\left\{ r  x + C(i-1,
r-1) \right\} \, f(x) \, \intd x.
\end{align}
Now we suppose that $C(i, j)$ has been computed for~$j<i$ and describe how to compute
$C(i, j+1)$. 
The analysis is similar as before, but in this case we have the additional option to reject an applicant and wait one more time step. 
The cost of waiting one step
is given by $C(i-1, j)$, so we get the following expression
\begin{align}
\label{eqn:dynamicj}
C(i, j+1) = \int_{0}^{\infty} \min\left\{ C(i-1, j), \min_{j+1 < r \le
i}\left\{ r x + C(i-1, r-1) \right\}\right\} \, f(x) \, \intd x.
\end{align}
If $C(i, j)$ has been computed for all $i \le n$ and all $j \le i$, then there
is a straightforward online algorithm that achieves expected cost $C(n, 0)$.
This algorithm simply waits for the cost $x$ of each applicant to be revealed,
and then chooses the action that minimizes the expression in the above
equations. 

\subsection{Analysis}

The computational efficiency of this algorithm depends on the difficulty of evaluating the
integrals in Equations~\eqref{eqn:dynamic0} and~\eqref{eqn:dynamicj}. 
For the simple case where the cost distributions are uniform, the right and side of both equations boild down to finding
the piecewise minimum over at most $n$ linear functions, which can easily be
computed. For other distributions, the algorithm may be slower. It is worth
noting that the algorithm cannot be applied in the case where the distribution
is unknown. For the case of a known distributions, we conclude the following.


\begin{theorem}
The dynamic program given by eqs.~\eqref{eqn:dynamic0} and~\eqref{eqn:dynamicj} yields an optimal online algorithm. 
\end{theorem}

Before we move on, we describe some shortcomings of this algorithm that we seek
to address in the remainder of this paper. 
The first issue of the algorithm is that, although it does provide an optimal competitive ratio, it is unclear how to analyze the algorithm, and in particular we do not know what competitive ratio
the algorithm guarantees. Secondly, the algorithm is very complicated to describe as it uses at least $n^2$ different threshold values to decide the hiring duration of an applicant, and these threshold values are specifically tailored to the distribution in question.
In the subsequent sections, we show that there exist algorithms with a constant competitive ratio, and in doing so we prove that the competitive ratio of the optimal online algorithm is also constant.
Thirdly, the optimal online algorithm requires both the cost distribution and the total number of time periods to be known ahead of time.
In contrast, in the following we develop an online algorithm with constant competitive ratio that still works even if neither information is known.

\section{Uniformly Distributed Costs}
\label{sec:uniform}

In this section, we give two algorithms with constant
competitive ratios in the case where applicants' costs are distributed uniformly. 
By shifting/rescaling we may assume without loss of generality that $X=\U[0,1]$, i.e., that the costs are distributed uniformly in the unit interval. Using Proposition~\ref{pro:opt}, we obtain the following expression for the expected cost of the offline optimum.
\begin{lemma}
$\E \bigl[\OPT_{n}\bigr] = \Hn_{n+1}-1$ for all $n \in \mathbb{N}$, where $\Hn_{n}$ is
the $n$-th harmonic number.
\label{lem:opt_uniform}
\end{lemma}
\begin{proof}{Proof.}
By Proposition~\ref{pro:opt}, 
$\E \bigl[\OPT_{n}\bigr] =  \sum_{i\in[n]}\int_{0}^{1}(1-x)^{i}\intd x 
  =  \sum_{i\in[n]}\frac{1}{i+1}=\Hn_{n+1}-1.$ \hfill \Halmos 
\end{proof}
\begin{algorithm}[tb]
\caption{A $\ratiouniformbad$-competitive algorithm for uniformly distributed costs.}
\label{alg:multi_uniform}
 $\ct \set 1$ \tcp*[r]{threshold cost}
 $t \set 1$ \tcp*[r]{remaining time with current threshold}
\For{$i \set 1, \dots, n$  }{ 
  $t \set t-1$\;
  \If{$x_i \le \ct$}{
     hire applicant $i$ for $\notnicefrac{4}{\ct}$ time steps\;
     \If{$i + \notnicefrac{4}{\ct} > n$}{
       {\bf stop}\;
     }
     $\ct \set \notnicefrac{\ct}{2}$; $t \set \notnicefrac{1}{\ct}$ 
  }

  \ElseIf{$t=0$}{
     $\ct \set 2\ct$; $t \set \notnicefrac{1}{\ct}$\;
  }
}
\end{algorithm}

\subsection{A First Competitive Algorithm}
\label{sec:uniform_first}
We start with our first online algorithm for uniform distributions (cf.~Algorithm~\ref{alg:multi_uniform}). 
The main idea of the algorithm is that whenever
we hire an applicant of cost $x$, we afterwards seek an applicant of cost $\notnicefrac{x}{2}$. The expected time until such an applicant arrives is $\notnicefrac{2}{x}$. 

If we set our hiring time equal to this expectation, we would
leave a considerable probability that we do not encounter any cheaper
applicants before the hiring time runs out. Instead, we hire the applicant
for $\notnicefrac{4}{x}$ steps and iteratively 
relax our hiring threshold after a certain time. 

More precisely, assume $x=\notnicefrac{1}{2^{j}}$ for some integer $j$. We then hire the
applicant for time
\begin{align}
\label{eq:geo_sum}
\frac{4}{x} > \frac{4}{x} - 1 = \frac{2}{x}+\frac{1}{x}+\frac{1}{2x}+\frac{1}{4x}+\dots+1.	
\end{align}
This way, if we do not find an applicant of cost at most $\notnicefrac{x}{2}$ during the next $\notnicefrac{2}{x}$
time steps,
we continue seeking for an applicant with cost $x$ for $\notnicefrac{1}{x}$ time steps, and so on. 
The geometric sum~\eqref{eq:geo_sum} just leaves enough time until we  eventually seek for an applicant with cost at most 1, who is surely found. 

To accommodate the fact that the costs of applicants are not powers of 2, in general, we maintain a 
threshold cost $\ct$ that is a power of 2 and reduce the threshold, whenever
a new applicant is hired, see Algorithm~\ref{alg:multi_uniform} for a formal description.
Finally, once an applicant is employed long
enough to cover all remaining time steps, we stop. Importantly, this
allows us to bound the lowest possible value of $\ct$ to be $2^{-\lceil\log n\rceil+2}$.\footnote{Here and throughout, we denote the logarithm of $n$ to base $2$ by $\log(n)$ and the natural logarithm of $n$ with $\ln(n)$.}
If an applicant is hired below this threshold, the hiring time is $\notnicefrac{4}{\ct}\geq n$.

In other words, during the course of the algorithm the threshold cost
$\ct$ can only take values of the form $2^{-j}$ for $j\in\{0,\dots,k-1\}$, where $k=\lceil\log(n)\rceil-1$.
This allows us to describe the evolution of $\ct$ with a Markov chain $M$
with $k+1$ states as follows. State $k$ is the absorbing
state that corresponds to the 
event that we \emph{succeeded} in hiring an applicant at cost at most the
threshold value 
$2^{-(k-1)}=2^{-\lceil\log n\rceil+2}$.
Each other state $j\in\{0,\dots,k-1\}$ corresponds to the event that the threshold value reaches $\ct=\ct_{j}:=2^{-j}$, see Figure~\ref{fig:markov_a} (a). 
Each transition of the Markov chain from a state $j$ to a state $j-1$ corresponds to the failure of finding an applicant below the threshold $\tau_j = 2^{-j}$ for $1/\tau_j = 2^j$ time steps, resulting in a doubling of the threshold cost. 
Each transition of the Markov chain from a state $j$ to a state $j+1$
corresponds to the hiring of an applicant resulting in the reduction of the threshold cost. We can therefore use the expected total number of state transitions of the Markov chain when starting at state $0$ to bound
the number of hired applicants overall.

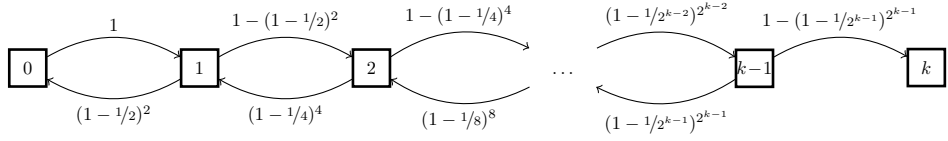
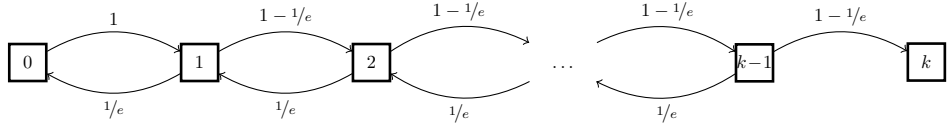
\begin{figure}[tb]
    \centering
\begin{subfigure}{\textwidth}
\centering
\begin{tikzpicture}[scale=0.3, every node/.style={scale=0.7}]
\matrix (m) [matrix of nodes, row sep=3em, column sep=5em,
	nodes={draw, rectangle, line width=1, text centered, minimum size=2em}]
    {\\  $0$ & $1$ & $2$  & |[draw=none]|$\phantom{1k}\dots\phantom{1k}$ & $\!\!k\!-\!1\!\!$  & $k$ \\ };

\draw[->] (m-2-1) edge [bend left] node [midway,above] {$1$} (m-2-2);
\draw[->] (m-2-2) edge [bend left] node [midway,below] {$(1-\nicefrac{1}{2})^2$} (m-2-1);

\draw[->] (m-2-2) edge [bend left] node [midway,above] {$1-(1-\nicefrac{1}{2})^2$} (m-2-3);
\draw[->] (m-2-3) edge [bend left] node [midway,below] {$(1-\nicefrac{1}{4})^4$} (m-2-2);

\draw[->] (m-2-3) edge [bend left] node [midway,above] {$1-(1-\nicefrac{1}{4})^4$} (m-2-4);
\draw[->] (m-2-4) edge [bend left] node [midway,below] {$(1-\nicefrac{1}{8})^8$} (m-2-3);

\draw[->] (m-2-4) edge [bend left] node [midway,above] {$(1-\nicefrac{1}{2^{k-2}})^{2^{k-2}}$} (m-2-5);
\draw[->] (m-2-5) edge [bend left] node [midway,below] {$(1-\nicefrac{1}{2^{k-1}})^{2^{k-1}}$} (m-2-4);

\draw[->] (m-2-5) edge [bend left] node [midway,above] {$1-(1-\nicefrac{1}{2^{k-1}})^{2^{k-1}}$} (m-2-6);


\foreach \x [count=\y] in {2,3}{
}

\end{tikzpicture}
\subcaption{Original Markov chain $M$.}
\end{subfigure}
~\\
\begin{subfigure}{\textwidth}
\centering
\begin{tikzpicture}[scale=0.3, every node/.style={scale=0.7}]
\matrix (m) [matrix of nodes, row sep=3em, column sep=5em,
	nodes={draw, rectangle, line width=1, text centered, minimum size=2em}]
    {\\  $0$ & $1$ & $2$  & |[draw=none]|$\phantom{1k}\dots\phantom{1k}$ & $\!\!k\!-\!1\!\!$  & $k$ \\ };

\draw[->] (m-2-1) edge [bend left] node [midway,above] {$1$} (m-2-2);
\draw[->] (m-2-2) edge [bend left] node [midway,below] {$\nicefrac{1}{\e}$} (m-2-1);

\draw[->] (m-2-2) edge [bend left] node [midway,above] {$1-\nicefrac{1}{\e}$} (m-2-3);
\draw[->] (m-2-3) edge [bend left] node [midway,below] {$\nicefrac{1}{\e}$} (m-2-2);

\draw[->] (m-2-3) edge [bend left] node [midway,above] {$1-\nicefrac{1}{\e}$} (m-2-4);
\draw[->] (m-2-4) edge [bend left] node [midway,below] {$\nicefrac{1}{\e}$} (m-2-3);

\draw[->] (m-2-4) edge [bend left] node [midway,above] {$1-\nicefrac{1}{\e}$} (m-2-5);
\draw[->] (m-2-5) edge [bend left] node [midway,below] {$\nicefrac{1}{\e}$} (m-2-4);

\draw[->] (m-2-5) edge [bend left] node [midway,above] {$1-\nicefrac{1}{\e}$} (m-2-6);


\foreach \x [count=\y] in {2,3}{
}

\end{tikzpicture}
\subcaption{Markov chain $\hat{M}(p,k)$ with homogeneous transition probabilities $p = 1-\nicefrac{1}{\e}$.}
\end{subfigure}

\caption{Markov chains $M$ modeling the expected number of hired applicants of
Algorithm~\ref{alg:multi_uniform}. Nodes correspond to states. State~$k = \lceil \log n \rceil-1$ is absorbing.\label{fig:markov_a}}
\end{figure}

Let $p_{j}$ denote the transition probability from state $j$ to
state $j+1$; i.e., when in state~$j$, the Markov chain transitions to state $j+1$ with probability $p_j$ and to state $j-1$ with probability $1-p_{j}$.
%
The probability that we fail to find an applicant with cost at most
$\tau$ during $\notnicefrac{1}{\tau}$ time steps is bounded
by
\[
1-p_{j} = (1-\tau)^{\notnicefrac{1}{\tau}}\leq \frac{1}{\e},
\]
i.e., $p_{j}\geq1-\nicefrac{1}{\e}$.
We set $p= 1-\nicefrac{1}{\e}$ and consider the Markov chain $\hat{M}(p,k)$ with homogeneous transition probability $p$ shown in Figure~\ref{fig:markov_a} (b). As we will show in the following lemma, the total number of state transitions to reach state $k$ in Markov chain $\hat{M}(p,k)$ provides an upper bound on the total number of state transitions to reach state $k$ in Markov chain $M$. The analysis of $\hat{M}(p,k)$ then yields the following result.

\begin{restatable}{lemma}{lemuniformhittingtime}
Starting in state $0$ of Markov chain $M$ with $k =\lceil\log(n) \rceil -1$, the expected number of state transitions is at most $\frac{\e k}{\e-2}$.\label{lem:uniform-hitting-time}
\end{restatable}

\begin{proof}{Proof.}
Let $k = \lceil \log n \rceil -1$ and $p = 1 - \nicefrac{1}{\e}$, and consider the Markov chain $\hat{M}(p,k)$ shown in Figure~\ref{fig:markov_a} (b). We first claim that the expected number of state transitions when starting in state~$0$ in Markov chain $M$ is bounded from above by that in Markov chain $\hat{M}(p,k)$. To see this, consider an arbitrary state~$j$ and consider the stochastic process that operates as $M$ with the exception that the first time state $j$ is visited, transition probabilities are as in $\hat{M}(p,k)$. Since $\hat{M}(p,k)$ has a higher probability to transition to a state with low index and the only absorbing state is $k$, this does not decrease the expected number of state transitions to state $j$ in $M$. 
Iterating this argument, we derive that also the stochastic process where state~$j$ \emph{always} transitions as in $\hat{M}(p,k)$ has a higher number of state transitions to state $j$. Iterating this argument over all states proves that the expected total number of state transitions in $M$ is upper bounded by the expected total number of state transitions in $\hat{M}(p,k)$.

In Lemma~\ref{lem:markov_a_visit} in \S~\ref{sec:markov_a}, we show that the expected number of visits for each state in $\hat{M}(p,k)$ is upper bounded by $\frac{1}{2p-1} = \frac{\e}{\e -2}$. Since we start in state $0$ and end after the first visit in state $k$, we conclude that the total number of state transitions of Markov chain $\hat{M}(p,k)$ is bounded by
\begin{align*}
\Bigl(\frac{\e}{\e-2}-1\Bigr) + \sum_{i=1}^{k-1}\frac{\e}{\e-2} + 1 = \frac{\e k}{\e-2}.
\end{align*}
This gives the claimed result. \hfill\Halmos
\end{proof}

We proceed to use Lemmas~\ref{lem:opt_uniform} and~\ref{lem:uniform-hitting-time} to obtain a first constant competitive algorithm for uniform costs.


\begin{restatable}{theorem}{thmuniformfirst}
Algorithm~\ref{alg:multi_uniform} is strictly $\ratiouniformbad$-competitive for uniform distributions.
\end{restatable}

\begin{proof}{Proof.}
Since $\ct$ decreases whenever an applicant is hired, we can bound
the number of hired applicants by the number of state transitions
from a state $j$ to state $j+1$ of the Markov chain. The algorithm
terminates at the latest when state $k=\lceil\log(n)\rceil-1$ is reached. If it ever
reaches that point, it has hired at least $k$ applicants and
every further hiring is mirrored by a state transition that decreases the current state. By using Lemma~\ref{lem:uniform-hitting-time} and only counting the transitions that increase the state index, we can
bound the expected number of hired applicants by
\[
\frac{\frac{\e k}{\e-2}-k}{2}+k=\left(\frac{\e}{\e-2}+1\right)\frac{k}{2}\leq\left(\frac{\e}{\e-2}+1\right)\frac{\log n}{2}.
\]
Whenever we hire an applicant below threshold $\tau$ the cost of the applicant is uniform in $[0,\tau]$, so the expected cost
is $\notnicefrac{\tau}{2}$. Since the hiring period is $\notnicefrac{4}{\tau}$ we get that  each hired applicant incurs an expected total cost of $2$. The threshold $\tau$ for the next candidate is independent of the exact cost of the last hire. Therefore we can combine the expected cost per candidate with Lemma~\ref{lem:opt_uniform} and we obtain 
\begin{align*}
\frac{\Exp{\ALG_n}}{\Exp{\OPT_n}} &\leq \frac{\ln n}{\Hn_{n+1}-1} \cdot \frac{1}{\ln 2}\left(\frac{\e}{\e-2}+1\right).
\intertext{Using Lemma~\ref{lem:ratiouniformbad} proven below where $\gamma$ is the Euler-Mascheroni constant this implies}
\frac{\Exp{\ALG_n}}{\Exp{\OPT_n}}  &\leq \left(1+\frac{20}{29}\left(\frac{5}{6}-\gamma\right)\right) \cdot \frac{1}{\ln 2}\left(\frac{\e}{\e-2}+1\right) < \ratiouniformbad,
\end{align*}
as claimed.
\hfill \Halmos
\end{proof}

\begin{lemma}
For any $n \in \mathbb{N}, \frac{\ln n}{\Hn_{n+1}-1} \leq 1+\frac{20}{29}\left(\frac{5}{6}-\gamma\right)$, where $\gamma$ is the Euler-Mascheroni constant. \label{lem:ratiouniformbad}
\end{lemma}
\begin{proof}{Proof.}
First, note that
\[\frac{\ln n}{\Hn_{n+1}-1} \leq \frac{\Hn_{n}-\gamma}{\Hn_{n+1}-1} = 1+\frac{1-\gamma-\frac{1}{n+1}}{\Hn_{n+1}-1}.\]
It suffices to prove that
\[\sup_{n \in \mathbb{N}}\frac{1-\gamma-\frac{1}{n+1}}{\Hn_{n+1}-1} = \sup_{n \in \mathbb{N}, n \geq 2}\frac{1-\gamma-\frac{1}{n}}{\Hn_{n}-1} \leq \frac{20}{29}\left(\frac{5}{6}-\gamma\right).\]
In order to do so, we show that there is a unique $n'\in \mathbb{N}_{\geq 2}$ with
\begin{align*}
\frac{1-\gamma-\frac{1}{n-1}}{\Hn_{n-1}-1} \leq \frac{1-\gamma-\frac{1}{n}}{\Hn_{n}-1} \quad &\text{for all } n \in \mathbb{N}_{\geq 2}, n \leq n'\text{, and}\\
\frac{1-\gamma-\frac{1}{n}}{\Hn_{n}-1} \geq \frac{1-\gamma-\frac{1}{n+1}}{\Hn_{n+1}-1} \quad &\text{for all } n \in \mathbb{N}_{\geq 2}, n \geq n',
\end{align*}
concluding that the supremum is attained at $n'$. Now we observe that
\begin{align*}
& & \frac{1-\gamma-\frac{1}{n}}{\Hn_{n}-1} - \frac{1-\gamma-\frac{1}{n+1}}{\Hn_{n+1}-1} &\geq 0\\
&\Leftrightarrow & \left(\Hn_{n+1}-1\right)\left(1-\gamma-\frac{1}{n}\right) - \left(\Hn_{n}-1\right)\left(1-\gamma-\frac{1}{n+1}\right) &\geq 0
\end{align*}
and
\begin{align*}
&\left(\Hn_{n+1}-1\right)\left(1-\gamma-\frac{1}{n}\right) - \left(\Hn_{n}-1\right)\left(1-\gamma-\frac{1}{n+1}\right)\\
&= \frac{1}{n+1}\left(1-\gamma-\frac{1}{n}\right) - \frac{1}{n}\left(\Hn_{n}-1\right) + \frac{1}{n+1}\left(\Hn_{n}-1\right)\\
&= \frac{1}{n+1}\left(1-\gamma-\frac{1}{n}\right) - \left(\Hn_{n}-1\right)\frac{1}{n(n+1)} = \frac{1}{n(n+1)}\left(n(1-\gamma)-\Hn_{n}\right),
\end{align*}
which is greater or equal to $0$ if and only if $n\geq 6$. We conclude that the supremum is attained at $n'=6$. We finish the proof by observing
\begin{align*}
\frac{\ln n}{\Hn_{n+1}-1} &\leq 1+ \frac{1-\gamma-\frac{1}{n+1}}{\Hn_{n+1}-1}\\ &\leq 1+ \sup_{\tilde{n} \in \mathbb{N}_{\geq 2}}\frac{1-\gamma-\frac{1}{\tilde{n}}}{\Hn_{\tilde{n}}-1}\\
&\leq 1+ \frac{1-\gamma-\frac{1}{6}}{\Hn_{6}-1}\\
&= 1+ \frac{20}{29}\left(\frac{5}{6}-\gamma\right),
\end{align*}
for all $n \in \mathbb{N}.$
\hfill \Halmos
\end{proof}

\subsection{Improving the Algorithm}\label{sec:uniform2}

We proceed to improve the competitive ratio of our algorithm as follows
(cf.~Algorithm~\ref{alg:multi_uniform2}). First, recall that, in Algorithm~\ref{alg:multi_uniform}, we hired an applicant below the current threshold of $\tau_j = 2^{-j}$ for $\notnicefrac{4}{\tau_j}$ time units with the rationale that
\begin{align*}
\sum_{i=0}^{j+1} \frac{1}{\tau_{i}}= \sum_{i=0}^{j+1} 2^{i} = 2^{j+2}-1 = \frac{4}{\tau_j} - 1 < \frac{4}{\tau_j}.
\end{align*}
With this inequality, it is ensured that we can afford $\notnicefrac{1}{\tau_{j+1}}$ time steps to look for an applicant below the threshold $\tau_{j+1}$ and, in case we did not find a suitable applicant, additional $\notnicefrac{1}{\tau_j}$ time steps looking for an applicant below the threshold $\tau_j$, and so on, until the threshold is raised to $1$ and we find a suitable applicant with probability $1$.

It turns out that it pays off to reduce both the hiring times and the time steps we spend looking for an applicant below a given threshold uniformly by a factor of $c := \notnicefrac{3}{4}$. That is, when hiring an applicant below the threshold of $\tau_j$, we hire only for $\notnicefrac{4c}{\tau_j} = \notnicefrac{3}{\tau_j}$ time units. To compensate for that we only look for an applicant below threshold $\tau_j$ for $\lceil \frac{c}{\tau_j} \rceil = \lceil \frac{3}{4\tau_j}\rceil$ time units. Note that for $\tau \in \{\notnicefrac{1}{2},1\}$ we round all times to the next integer. For $j \geq 3$, we then obtain
\begin{align*}
\sum_{i=0}^{j+1} \biggl\lceil \frac{3}{4\tau_j} \biggr\rceil = \biggl\lceil \frac{3}{4} \biggr\rceil	+ \biggl\lceil \frac{3}{2} \biggr\rceil + \sum_{i=2}^{j+1} \frac{3}{4\tau_j} = 1 + 2 + 3 \sum_{i=2}^{j+1} 2^{i-2} = 3 \cdot 2^j = \frac{3}{\tau_j} \;.
\end{align*}
Similarly, we may check for $j=0$ that $\lceil \notnicefrac{3}{4} \rceil + \lceil \notnicefrac{3}{2} \rceil = 3 = \notnicefrac{3}{\tau_0}$, and for $j=1$ that $\lceil \notnicefrac{3}{4} \rceil + \lceil \notnicefrac{3}{2} \rceil +3 = 6 = \notnicefrac{3}{\tau_1}$. Thus, we may conclude that the above choices ensure that an applicant is under contract at all times.

Second, instead of reducing the threshold
once by factor 2 when we hire a new applicant, we repeatedly halve
the threshold for as long as it is still greater or equal to the actual cost of the new applicant. This
way, we can ensure that the cost for which a new applicant is hired
is always uniformly distributed in $[\ct,2\ct)$ (or possibly $[0,2\ct)$
for the last hiring), where $\tau$ denotes the threshold \emph{after} the applicant is hired. Thus, the expected total cost of each applicant
is $\frac{3\tau}{2} \cdot \frac{2c}{\ct}=3c$ (or possibly $2c$ for
the last hiring).

\begin{algorithm}[tb]
\caption{A $\ratiouniform$-competitive algorithm for uniformly distributed costs.\label{alg:multi_uniform2}
}
 $\ct \set 1$ \tcp*[r]{threshold cost}
 $t \set 1$ \tcp*[r]{time with threshold}
\For{$i \set 1, \dots, n$  }{ 
  $t \set t-1$\;
  \If{$x_i \le \ct$}{
     \While{$x_i \le \ct$}{
       $\ct \set \notnicefrac{\ct}{2}$ \vphantom{$\delta_q$}\;
     }
     hire applicant $i$ for $\lceil \nicefrac{2c}{\ct}\rceil$ time steps$\vphantom{\delta_q}$\;
     \If{$i +  \lceil \nicefrac{2c}{\ct}\rceil > n$$\vphantom{\delta_q}$}{
       {\bf stop}\;
     }
     $t \set \lceil \notnicefrac{c}{\ct} \rceil$ $\vphantom{\delta_q}$
  }

  \ElseIf{$t=0$}{
     $\ct \set 2\ct$; $t \set \lceil \notnicefrac{c}{\ct} \rceil$ \vphantom{$\delta_q$}
  }
}
\end{algorithm}

Once we hire an applicant with a cost below $2^{-j}$, the threshold $\tau$ after hiring is at most $2^{-(j+1)}$ so that the applicant is hired for at least $\lceil 2\frac{c}{\tau} \rceil \geq 2^{j+2}c$ time steps. This implies that we only need to account for thresholds of the form $\ct_j = 2^{-j}$ where $j\in\{0,1,\dots, \lceil\log(\notnicefrac{n}{c})\rceil-2\}$.
We again capture the behavior of the algorithm with a Markov chain
(cf.~Figure~\ref{fig:markov_b}). To this end, states $A_{0},A_{1},\dots,A_{k}$
and $B_{0},B_{1},\dots,B_{k}$ with $k=\lceil\log(\notnicefrac{n}{c})\rceil-2$ are introduced. We distinguish
between the states $A_{j}$ that correspond to the algorithm looking
for suitable applicants by comparing their cost with $\ct_j = 2^{-j}$, and states
$B_{j}$ that correspond to the event that the cost of our current candidate 
is below the threshold $\ct_j = 2^{-j}$.
Each state $A_{j}$ with $j > 0$ either transitions to $A_{j-1}$ with probability $(1-p_j)$, when no
applicant for the current threshold was found, or to $B_{j}$ with
probability $p_j$. As for the previous Markov chain, we have
\[
(1-p_j)=(1-\tau_j)^{\lceil c/\tau_j \rceil} \leq \e^{-c}.
\]
Similar to the previous section, we may consider the Markov chain with homogenous transition probabilities $p=1-\e^{-c}$ shown in Figure~\ref{fig:markov_b} (b) instead, since we are only interested in
upper bounding the number of hired applicants. Each state $B_{j}$ with $j < k$ transitions
to $B_{j+1}$ or $A_{j+1}$ each with probability $\notnicefrac{1}{2}$,
since the cost $x$ lies with equal probability in $[\tau,2\tau)$
or $[0,\tau)$. State $B_{k}$ is the only absorbing state of the
Markov chain. Our analysis of the Markov chain in \S~\ref{sec:markov_b} yields the following result.

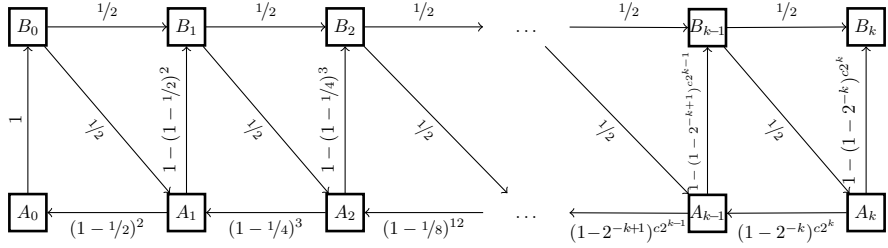
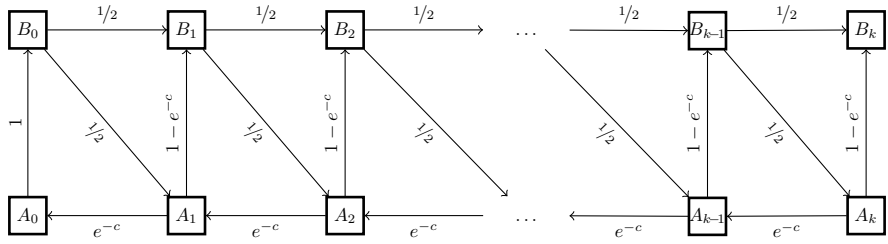
\begin{figure}[bt]
\begin{subfigure}{\textwidth}
    \centering
\begin{tikzpicture}[scale=0.3, every node/.style={scale=0.7}]
\matrix (m) [matrix of nodes, row sep=5.5em, column sep=4.5em,
	nodes={draw, rectangle, line width=1, text centered, minimum size=2em}]
    { $B_0$ & $B_1$ & $B_2$ &|[draw=none]| $\phantom{B_1}\dots\phantom{B_1}$  & $\!\!\!B_{k\!-\!1}\!\!\!$ & $B_{k}$ \\
      $A_0$ & $A_1$ & $A_2$ &|[draw=none]| $\phantom{A_1}\dots\phantom{A_1}$  & $\!\!\!A_{k\!-\!1}\!\!\!$  & $A_{k}$ \\ };
\draw[->] (m-2-1) -- (m-1-1) node [above,midway,sloped] {$1$};
\draw[->] (m-2-2) -- (m-1-2) node [above,midway,sloped] {$1-(1-\nicefrac{1}{2})^{2}$};
\draw[->] (m-2-3) -- (m-1-3) node [above,midway,sloped] {$1-(1-\nicefrac{1}{4})^{3}$};
\draw[->] (m-2-5) -- (m-1-5) node [above,midway,sloped,rotate=180] {\scriptsize$1-(1-2^{-k+1})^{c2^{k-1}}$};
\draw[->] (m-2-6) -- (m-1-6) node [above,midway,sloped,rotate=180] {$1-(1-2^{-k})^{c2^{k}}$};

\draw[->] (m-2-2) -- (m-2-1) node [below,midway] {$(1-\nicefrac{1}{2})^{2}$};
\draw[->] (m-2-3) -- (m-2-2) node [below,midway] {$(1-\nicefrac{1}{4})^{3}$};
\draw[->] (m-2-4) -- (m-2-3) node [below,midway] {$(1-\nicefrac{1}{8})^{12}$};
\draw[->] (m-2-5) -- (m-2-4) node [below,midway] {$(1\!-\!2^{-k\!+\!1})^{c2^{k\!-\!1}}$};
\draw[->] (m-2-6) -- (m-2-5) node [below,midway] {$(1-2^{-k})^{c2^{k}}$};

\foreach \x [count=\y] in {2,3,4,5,6}{
    \draw[->] (m-1-\y) -- (m-1-\x) node [above,midway,sloped] {$\nicefrac{1}{2}$};
    \draw[->] (m-1-\y) -- (m-2-\x) node [below,midway,sloped] {$\nicefrac{1}{2}$};
}

\end{tikzpicture}
\subcaption{Original Markov chain $N$.}
~\\
~\\
\end{subfigure}
\begin{subfigure}{\textwidth}
    \centering
\begin{tikzpicture}[scale=0.3, every node/.style={scale=0.7}]
\matrix (m) [matrix of nodes, row sep=5.5em, column sep=4.5em,
	nodes={draw, rectangle, line width=1, text centered, minimum size=2em}]
    { $B_0$ & $B_1$ & $B_2$ &|[draw=none]| $\phantom{B_1}\dots\phantom{B_1}$  & $\!\!\!B_{k\!-\!1}\!\!\!$ & $B_{k}$ \\
      $A_0$ & $A_1$ & $A_2$ &|[draw=none]| $\phantom{A_1}\dots\phantom{A_1}$  & $\!\!\!A_{k\!-\!1}\!\!\!$  & $A_{k}$ \\ };
\draw[->] (m-2-1) -- (m-1-1) node [above,midway,sloped] {$1$};
\draw[->] (m-2-2) -- (m-1-2) node [above,midway,sloped] {$1-e^{-c}$};
\draw[->] (m-2-3) -- (m-1-3) node [above,midway,sloped] {$1-e^{-c}$};
\draw[->] (m-2-5) -- (m-1-5) node [above,midway,sloped,rotate=180] {$1-e^{-c}$};
\draw[->] (m-2-6) -- (m-1-6) node [above,midway,sloped,rotate=180] {$1-e^{-c}$};

\draw[->] (m-2-2) -- (m-2-1) node [below,midway] {$e^{-c}$};
\draw[->] (m-2-3) -- (m-2-2) node [below,midway] {$e^{-c}$};
\draw[->] (m-2-4) -- (m-2-3) node [below,midway] {$e^{-c}$};
\draw[->] (m-2-5) -- (m-2-4) node [below,midway] {$e^{-c}$};
\draw[->] (m-2-6) -- (m-2-5) node [below,midway] {$e^{-c}$};

\foreach \x [count=\y] in {2,3,4,5,6}{
    \draw[->] (m-1-\y) -- (m-1-\x) node [above,midway,sloped] {$\nicefrac{1}{2}$};
    \draw[->] (m-1-\y) -- (m-2-\x) node [below,midway,sloped] {$\nicefrac{1}{2}$};
}

\end{tikzpicture}
\subcaption{Markov chain $\hat{N}(p,k)$ with homogenous transition probabilities $p = 1- e^{-c}$.}
\end{subfigure}
\caption{Markov chains modeling the expected number of hired applicants of
Algorithm~\ref{alg:multi_uniform2}.\label{fig:markov_b}}

\end{figure}

\begin{restatable}{lemma}{lemuniformhittingtimesecond}
Starting in state $A_0$ of Markov chain $N$, the expected number of transitions from an $A$-state to a $B$-state is at most
\label{lem:uniform_hitting_time_2}
\begin{align}
\label{eq:h}
h = \frac{kp}{3p-1} -\frac{4p(1-2p)}{(3p-1)^{2}} + \left(\frac{1-p}{3p-1}\right)^{\!\!2}\!\! \left( \frac{2(1-p)}{1+p}\right)^{\!\!k}
\end{align}
where $k = \lceil \log (n/c) \rceil -2$ and $p = 1-e^{-c}$.
\end{restatable}

\begin{proof}
Let $c= 3/4$, $k = \lceil \log(n/c) \rceil -2$, and $p = 1 - \e^{-c}$. We again argue that we only overestimate the expected visiting times when considering Markov chain $\hat{N}(p,k)$ instead of $N$. 
To see this, fix a state~$A_j$, $j=0,\dots,k$ and consider the stochastic process $N'$ that follows Markov chain $N$, but, the first time state~$A_j$ is visited, transitions according to the probabilities of $\hat{N}(p,k)$. 
As in all stochastic processes we consider the expected number of visits of all states is decreasing in the index of the starting state $A_j$, the expected number of visits of all states are not smaller in $N'$ than in $N$. Iterating this argument, we conclude that the expected number of visits of all states in $N$ does not exceed those in $\hat{N}(p,k)$.

In Lemma~\ref{lem:markov_b_transitions1} in \S~\ref{sec:markov_b} we prove that the expected number of transitions from an $A$-state to a $B$-state of $\hat{N}(p,k)$ is bounded from above by \eqref{eq:h}.\hfill \Halmos
\end{proof}

As every transition from an $A$-state to a $B$-state corresponds to the hiring of a candidate, bounding these transitions allows us to bound $\Ex{\ALG_n}$. Together with the formula for $\Ex{\OPT_n}$ proven in Lemma~\ref{lem:opt_uniform} we obtain an improved competitive ratio. Numerically, the choice $c=\notnicefrac{3}{4}$ yields optimizes the competitive ratio yielding a strict competitive ratio of $\ratiouniform$.

\begin{restatable}{theorem}{thmuniformimproved}
\label{thm:uniform_upper}
For $c=\notnicefrac{3}{4}$, Algorithm~\ref{alg:multi_uniform2} is strictly
$\ratiouniform$-competitive for uniform distributions.
\end{restatable}

\begin{proof}{Proof.}
Whenever an applicant is hired, the Markov chain transitions from $A_{j}$ to $B_{j}$ for some value $j\in[k]$. The algorithm terminates at the latest when state $B_{k}$ is reached. We can thus bound the number of hired applicants by the expression $h$ of Lemma~\ref{lem:uniform_hitting_time_2}. Using Lemma~\ref{lem:opt_uniform} and the fact that the expected cost incurred by each hired applicant is $3c$ (and $2c$ for the last hiring), we get
\[
\frac{\Exp{\ALG_n}}{\Exp{\OPT_n}}\leq\frac{3hc-c}{\Hn_{n+1}-1}\leq \ratiouniform,
\]
for $c=\notnicefrac{3}{4}$ and all $n$. See Lemma~\ref{lem:uniform_2.98} below for a proof of the last inequality. \hfill \Halmos
\end{proof}

\newcommand{\lowern}{192}

\begin{lemma}
\label{lem:uniform_2.98}
Let $c\!=\!\notnicefrac{3}{4}$, $p\!=\!1-e^{-c}$, $k\!=\!\lceil \log(\notnicefrac{n}{c})\rceil-2$, and 
$h=\frac{kp}{3p-1} -\frac{4p(1-2p)}{(3p-1)^{2}} + \bigl(\frac{1-p}{3p-1}\bigr)^2 \bigl( \frac{2(1-p)}{1+p}\bigr)^{\!k}$. Then, 
$\frac{3hc-c}{\Hn_{n+1}-1} \leq \ratiouniform$ for all $n$.
\end{lemma}


\begin{proof}{Proof.}
Since the expression $3hc - c$ is constant as long as $k= \lceil \log (\notnicefrac{n}{c}) \rceil - 2 = \lceil \log(\notnicefrac{n}{3}) \rceil$ is constant, the ratio $\frac{3hc - c}{\Hn_{n+1}-1}$ is maximized for some $n$ of the form $n = 3\cdot 2^{\ell-1}+1$ with $\ell \in \mathbb{N}$. See also Figure~\ref{fig:plots} where the ratio is plotted as a function of $n$. The claim of the lemma is easily verified for $\ell =1,\dots,6$. For $\ell \geq 7$, we obtain
\begin{align}
h &=  \frac{\ell p}{3p-1} -\frac{4p(1-2p)}{(3p-1)^{2}} + \left(\frac{1-p}{3p-1}\right)^{\!\!2} \cdot\left( \frac{2(1-p)}{1+p}\right)^{\!\!\ell} \notag \\[6pt]
   &\le \frac{\ell p}{3p-1} -\frac{4p(1-2p)}{(3p-1)^{2}} + \left(\frac{1-p}{3p-1}\right)^{\!\!2} \cdot\left( \frac{2(1-p)}{1+p}\right)^{\!\!7} \label{eq:put_in_p}\\[6pt]
   &< \frac{\ell p}{3p-1} + \frac{1}{\e}, \notag
\end{align}
where for the first inequality we used that for $p = 1 - e^{-3/4}$ we have $2(1-p)/(1+p) = 2/(2\e^{3/4}-1)\approx 0.618 < 1$ and for the second inequality we evaluated \eqref{eq:put_in_p} for $p= 1- e^{-3/4}$.
%
For the Euler-Mascheroni constant $\gamma \approx 0.577$ we obtain 
\begin{align}
   \frac{3hc-c}{\Hn_{3\cdot 2^{\ell-1}+2}-1} 
   &< 
   \frac{\frac{9}{4} \bigl(\frac{\ell p}{3p-1} + \frac{1}{\e}\bigr)-\frac{3}{4}}
          {\ln(3\cdot 2^{\ell-1}+2) - (1 - \gamma)} \notag\\[6pt]
          &\leq \frac{\frac{9}{4} \bigl(\frac{\ell p}{3p-1} + \frac{1}{\e}\bigr)-\frac{3}{4}}
          {\ell \ln(2) + \ln(3) - \ln(2) - (1 - \gamma)},\label{eq:complicated_expression}
   \end{align}
where we used that the denominator is positive.
Using that $\ln(3) - \ln(2) - (1- \gamma) \approx -0.017 < 0$, elementary calculus shows that the expression in \eqref{eq:complicated_expression} is decreasing in $\ell$. 
Evaluating it for $\ell=7$ we obtain
\begin{align*}
  \frac{3hc-c}{\Hn_{3\cdot 2^{\ell-1}+2}-1} \leq \ratiouniform,
 \end{align*}
 as claimed. \hfill\Halmos
\end{proof}


\subsection{Analytical Lower Bound}
\label{sec:lower_bound}

To obtain a lower bound on the competitive ratio of any online algorithm, we study in this section a relaxation of the problem. The relaxation allows to exploit an interesting connection to the classical stopping problem with uniformly distributed random variables which is known under the name \emph{Cayley-Moser problem}, see \citet{moser1956,GilbertMosteller/66}.

Consider the relaxation where we are allowed to hire an applicant for \emph{any} (not necessarily contiguous) subset of all future time steps, while still having to decide on this set immediately upon arrival of the applicant.
In this setting, there is obviously no advantage of concurrent employment --- once we hired an applicant for some time slot, there is no benefit of hiring additional applicants for the same time slot.
Put differently, the decision whether to hire an applicant for some time slot is independent of the decision for other time slots.
Thus, the problem reduces to simultaneously solve a stopping problem for each time slot $t$. That is, we need to hire exactly one of the first $t$ applicants for this time slot, while applicants appear one by one and we need to irrevocably hire or discard each applicant upon their arrival. By linearity of expectation, each of the $n$ stopping problems can be treated individually.

Gilbert and Mosteller~\cite{GilbertMosteller/66} showed that, in the maximization version of the single stopping problem with uniformly distributed values, the optimal stopping rule is a threshold rule parametrized by $t$ thresholds $\tau_0,\dots,\tau_{t-1}$. The rule stops at a time step when there are $i$ remaining (unobserved) random variables and the realization is above $\tau_i$. The threshold values follow the recursion $\tau_0 = 0$ and $\tau_{i+1} = (1+\tau_{i}^2)/2$ for all $i \geq 1$. Gilbert and Mosteller showed that the value $\tau_t$ is also the expected revenue for the stopping problem with $t$ slots.when following the optimal strategy. They bound the expected revenue for all $t$ by
$$\tau_t \geq 1 - \frac{2}{t + \ln(t+1) + 1.767}.$$
By symmetry of the uniform distribution, for the corresponding single stopping problem with uniformly distributed costs and minimization objective, this immediately yields that the optimum expected cost~$1-\tau_t$ is lower bounded by
$h(t) := \frac{2}{t + \ln(t+1) + 1.767}$.

Since we need to solve a stopping problem for each time slot $1,2,\dots,n$, and by linearity of expectation, we get a lower bound on the expected cost of $\sum_{t=1}^{n} h(t)$ for the relaxed problem.
On the other hand, by Lemma~\ref{lem:opt_uniform}, for the offline optimum of our original problem, we have $\E_{x_1 \sim \U[0,1],\dots,x_n \sim \U[0,1]} \bigl[\OPT_{n}\bigr] = \sum_{t=1}^{n} g(t) := \sum_{t=1}^{n} \frac{1}{1+t}$.
Since $h(t)$ and $g(t)$ are both monotonically decreasing, we can estimate $\sum_{t=1}^{n} h(t) \geq \int_1^{n+1} h(t) \,\mathrm{d}t$ and $\sum_{t=1}^{n} g(t) \leq \int_0^{n} g(t) \,\mathrm{d}t$.
Also, since both integrals tend to infinity for growing~$n$, we can apply
l'H\^opital's rule and obtain
\begin{align*}
 \lim_{n \to \infty} \frac{\sum_{t=1}^{n} h(t)}{\sum_{t=1}^{n} g(t)} 
&\geq \lim_{n \to \infty} \frac{\int_{1}^{n+1} h(t) \,\mathrm{d}t}{\int_{0}^{n} g(t) \,\mathrm{d}t} \\
&= \lim_{n \to \infty} \frac{h(n+1)}{g(n)} \\
&= \lim_{n \to \infty} \frac{2(n+1)}{n+1+\log(n+2)+1.767} = 2. 
\end{align*}

As $\sum_{t=1}^{n} h(t)$ is the expected cost of an optimum online solution to the relaxed problem, it is a lower bound on the expected cost of an optimum online solution to the original problem, and we get the following bound.

\begin{theorem}
	Asymptotically, for a uniform distribution, no online algorithm has a competitive ratio below~$2$.
\end{theorem}

A plot of the lower bound $\frac{\sum_{t=1}^n h(t)}{H_{n+1}-1}$ as a function of $n$ shown in Figure~\ref{fig:plots} reveals that the lower bound converges very slowly. Even for $n=10,000$, the lower bound is still below $9/5$.

\subsection{Computational Lower Bound}

In this section, we give a computational lower bound based on an optimal online algorithm for uniformly distributed costs. This gives a slightly higher lower bound than the analytical bound above. We implemented the
optimal online algorithm presented in \S~\ref{sec:dp} in exact arithmetic, using
rounding to prevent numbers from getting too large. The algorithm achieves a
competitive ratio of 2.148 for an instance with 10,000 time steps, see Figure~\ref{fig:plots}. We describe the details on the computational lower bound in the following paragraphs. 

\begin{figure}
\begin{subfigure}[c]{0.495\textwidth}
\centering
\includegraphics[scale=0.65]{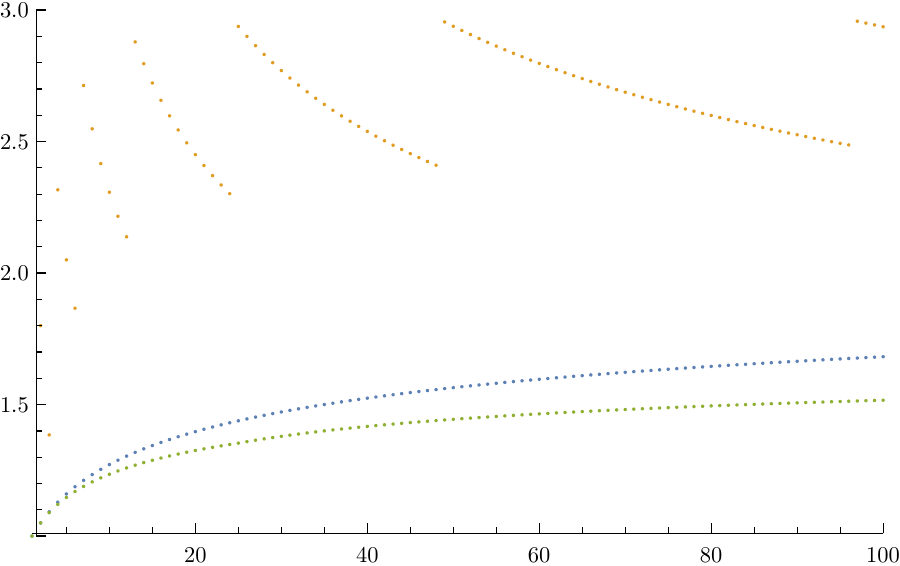}
\subcaption{$n=1,\dots,100$}
\end{subfigure}
\begin{subfigure}[c]{0.495\textwidth}
\centering
\psfrag{2}{$2000$}
\includegraphics[scale=0.65]{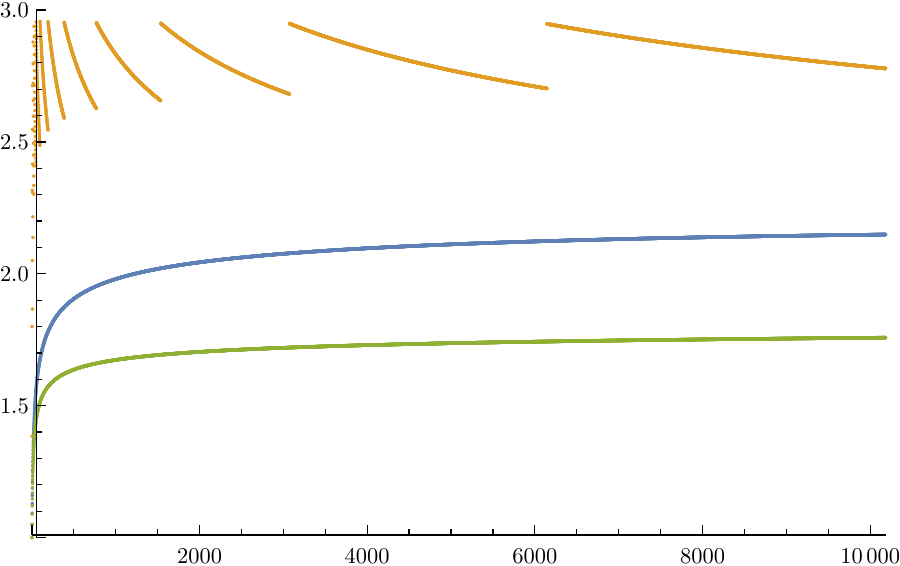}
\subcaption{$n=1,\dots,10,000$}
\end{subfigure}
\caption{Competitive ratio of our algorithm (orange), and the optimal online algorithm (blue) for uniformly distributed cost. The lower bound via the Gilbert-Mosteller problem is shown in green.\label{fig:plots}}
\end{figure}

For the uniform distribution, we know from Lemma~\ref{lem:opt_uniform} that the optimal offline
algorithm has expected cost $\mathcal{H}_{n+1} - 1$. On the other hand, the entry $C(n, 0)$ in the dynamic programming table
of the optimal online algorithm gives the optimal cost for an instance with $n$
days.
 Therefore, for every $n >
0$ the ratio $\frac{C(n, 0)}{\mathcal{H}_{n+1} - 1}$ provides a lower bound on the
best strict competitive ratio achievable by any online algorithm.

We implemented the optimal online algorithm for
the uniform case, and computed the expression above for increasing
values of $n$. 
In order to obtain a conclusive proof, one needs to implement the
algorithm in exact rational arithmetic. However, in doing so, we found that the size of the numerators and denominators grow very quickly in $n$, and already for $n=22$
both the numerator and the denominator have over a million digits. This makes it computationally
intractable to compute $\frac{C(n, 0)}{\mathcal{H}_{n+1} - 1}$ for large $n$.

To address this, we adopted a rounding scheme: after computing $C(i, j)$ for some
$i$ and $j$, we rounded the number down to another rational with a smaller
numerator and denominator, and then stored the rounded number in the dynamic
programming table. Since we only ever round down, the resulting costs computed
by the algorithm must always be cheaper than the expected cost of the optimal
online algorithm. Therefore, the computed value of $\frac{C(n,
0)}{\mathcal{H}_{n+1} - 1}$ is still a lower bound on the strict competitive
ratio that can be achieved. 

Ultimately, we found that for $n=10,000$ the competitive ratio can be no better
than~$2.148$.
The following theorem summarizes the results.

\begin{restatable}{theorem}{thmlowerbound}
    For a uniform distribution, no online algorithm has  strict competitive ratio below~$2.148$.\label{thm:lower_bound}
\end{restatable}


\section{Arbitrary Distributions}\label{arbitrary_distributions}

\begin{algorithm}[tb]

\caption{A $\ratioarbitrary$ competitive algorithm for arbitrary distributions.\label{alg:arb_dist}}
 $q \set 1$ \tcp*[r]{threshold quantile index}
 $\ct \set \delta_q$ \tcp*[r]{threshold cost}
 $t \set 1$ \tcp*[r]{time with threshold}
\For{$i \set 1, \dots, n$  }{ 
  $t \set t-1$\,\,\;
  \If{$x_i \le \ct$}{
     \While{$x_i \le \ct$}{
       $q \set \notnicefrac{q}{2}$; $\ct \set \delta_q$\,\,\;
     }
	 \If{$i + \nicefrac{2}{q} > n$}{
	   hire applicant~$i$ until time $n$\,\,\;
       {\bf stop}\,\,\;
     }  
     hire applicant $i$ for $\notnicefrac{2}{q}$ time steps\,\,\;
     $t \set \notnicefrac{1}{q}$\,\,\;
  }

  \ElseIf{$t=0$}{
     $q \set 2q$; $\ct \set \delta_q$; $t \set \notnicefrac{1}{q}$\,\,\;
  }
}
\end{algorithm}

In this section, we generalize Algorithm~\ref{alg:multi_uniform2} for the choice of $c=1$ to an arbitrary
distribution $X$ (cf.~Algorithm~\ref{alg:arb_dist}). Whenever
we halve our threshold in the course of Algorithm~\ref{alg:multi_uniform2},
we essentially halve the probability mass of $X$ below the
threshold (i.e., the probability that a drawn value lies below $\ct$).
To achieve the same effect with respect to an arbitrary distribution $X$, we consider \emph{quantiles}
$\delta_{q}$ of $X$, defined by the property that $\Prob{x\leq \delta_q}=q$ for continuous
distributions\footnote{In general, we need to define $\delta_{q}$ more carefully
via $\Prob{x\leq \delta_q}\geq q$ and $\Prob{x\geq\delta_{q}}\geq1-q$.}.
Algorithm~\ref{alg:arb_dist} changes the threshold by halving
and doubling $q$ and using $\ct=\delta_{q}$, which results in the
same behavior as Algorithm~\ref{alg:multi_uniform2} when $X$
is uniform. Therefore we can in principle analyze the algorithm for general distributions using a Markov chain similar to that in \S~\ref{sec:uniform2}. Specifically, the Markov chain again governs the evolution of the value $q = 2^{-j}$, $j=0,1,2,\dots$ and the corresponding threshold value $\tau = \delta_q$ in the course of Algorithm~\ref{alg:arb_dist}. After finding an applicant with a cost~$x_i$ below the threshold value $\tau$, the value of $q$ is halved until $\delta_q < x_i \leq \delta_{q-1}$. Since the applicant is then hired for $2/q$ time steps where $q$ is the value after the halving, we conclude that when hiring an applicant below the threshold of $\delta_q$ it is hired for $4/q$ time units. 
As the process stops at the latest when $4/q \geq n$, the Markov chain~$N'$ has states $A_0, A_1, \dots, A_k$ and $B_0, B_1, \dots, B_k$ with $k=\lceil\log n\rceil-2$, see Figure~\ref{fig:markov_arbitrary}.

\begin{figure}[bt]
    \centering
\begin{tikzpicture}[scale=0.3, every node/.style={scale=0.7}]
\matrix (m) [matrix of nodes, row sep=5.5em, column sep=4.5em,
	nodes={draw, rectangle, line width=1, text centered, minimum size=2em}]
    { $B_0$ & $B_1$ & $B_2$ &|[draw=none]| $\phantom{B_1}\dots\phantom{B_1}$  & $\!\!\!B_{k\!-\!1}\!\!\!$ & $B_{k}$ \\
      $A_0$ & $A_1$ & $A_2$ &|[draw=none]| $\phantom{A_1}\dots\phantom{A_1}$  & $\!\!\!A_{k\!-\!1}\!\!\!$  & $A_{k}$ \\ };
\draw[->] (m-2-1) -- (m-1-1) node [above,midway,sloped] {$1$};
\draw[->] (m-2-2) -- (m-1-2) node [above,midway,sloped] {$1-(1-\nicefrac{1}{2})^{2}$};
\draw[->] (m-2-3) -- (m-1-3) node [above,midway,sloped] {$1-(1-\nicefrac{1}{4})^{4}$};
\draw[->] (m-2-5) -- (m-1-5) node [above,midway,sloped,rotate=180] {\scriptsize$1-(1-2^{-k+1})^{2^{k-1}}$};
\draw[->] (m-2-6) -- (m-1-6) node [above,midway,sloped,rotate=180] {$1-(1-2^{-k})^{2^{k}}$};

\draw[->] (m-2-2) -- (m-2-1) node [below,midway] {$(1-\nicefrac{1}{2})^{2}$};
\draw[->] (m-2-3) -- (m-2-2) node [below,midway] {$(1-\nicefrac{1}{4})^{4}$};
\draw[->] (m-2-4) -- (m-2-3) node [below,midway] {$(1-\nicefrac{1}{8})^{8}$};
\draw[->] (m-2-5) -- (m-2-4) node [below,midway] {$(1\!-\!2^{-k\!+\!1})^{2^{k\!-\!1}}$};
\draw[->] (m-2-6) -- (m-2-5) node [below,midway] {$(1-2^{-k})^{2^{k}}$};

\foreach \x [count=\y] in {2,3,4,5,6}{
    \draw[->] (m-1-\y) -- (m-1-\x) node [above,midway,sloped] {$\nicefrac{1}{2}$};
    \draw[->] (m-1-\y) -- (m-2-\x) node [below,midway,sloped] {$\nicefrac{1}{2}$};
}

\end{tikzpicture}
\caption{Markov chain $N'$ modeling the expected number of hired applicants of
Algorithm~\ref{alg:arb_dist}.\label{fig:markov_arbitrary}}
\end{figure}
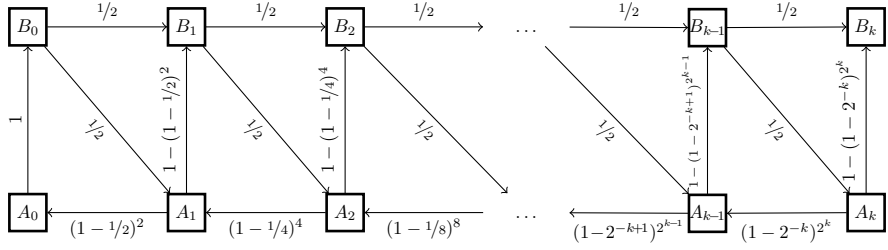

Again, we start in state $A_0$, $B_k$ is the absorbing state, and states $A_j,B_j$ correspond to states of the algorithm where $\ct=\delta_{2^{-j}}$. The transition probability from any state $A_j$ to state 
$B_j$ is bounded from below by $p=1-\nicefrac{1}{\e}$, since the probability of not finding any applicant of cost at most~$\delta_q$ within $1/q$ steps is
\begin{align*}
1-\P[x>\delta_{q}]^{1/q} = 1-(1-\P[x\leq\delta_{q}])^{1/q} \geq 1-(1-q)^{1/q} \geq 1-\nicefrac{1}{\e}.
\end{align*}


The analysis of the algorithm for arbitrary distributions, however, turns out to be more intricate than for the uniform case for two main reasons. 
First, for uniform distributions it was sufficient to count the total number of transitions from an $A$-state to a $B$-state as any such transition corresponds to the hiring of a candidate with a total cost of $2$. 
On the other hand, for general distributions we need to bound the number of transitions from state $B_j$ to state $A_{j+1}$ for each $j$ individually, as the resulting costs may differ among the different values of $j$. 
The following lemma provides a bound independent of~$j$.

\begin{restatable}{lemma}{lemuniformtransitions}
Starting in state $A_0$ of Markov chain $N'$, for each $j\in\{0, \ldots, k-1\}$ the expected number of transitions from $B_j$ to $A_{j+1}$ is at most $\frac{p}{3p-1}$, where $p = 1-\nicefrac{1}{\e}$.
\label{lem:uniform_transitions}
\end{restatable}

\begin{proof}{Proof.}
With the same arguments as in the proof of Lemma~\ref{lem:uniform_hitting_time_2}, we obtain an upper bound on the expected number of transitions from $B_j$ to $A_{j+1}$ by considering the Markov chain $\hat{N}(p,k)$ with homogenous transition probability $p = 1 - \nicefrac{1}{\e}$ and $k = \lceil \log n \rceil -2$. For the latter Markov chain, Lemma~\ref{lem:markov_b_transitions2} proven in \S~\ref{sec:markov_b} establishes the result.\hfill\Halmos	
\end{proof}

The second main issue when analyzing the competitive ratio of the algorithm is the lack of a concrete value for $\OPT_n$ for general distributions. Thus, we need the following lemma that expresses $\E\bigl[\OPT_n]$ as a sum over conditional expectations of the form $\E\bigl[ x \,|\, \delta_{2^{-(r+1)}} < x \leq \delta_{2^{-r}}\bigr]$. 

\begin{restatable}{lemma}{lemarbdistofflineopt}
\label{lem:arb_dist_offline_opt}
Let $n \in \mathbb{N}$, $k \!=\! \lceil \log n\rceil \!-\! 2$ and $\eta := \frac{5}{2} - \frac{55}{6\e^2} \approx 1.259$.
Then, we have
$$
\E\bigl[\OPT_n \bigr]
\geq \sum_{r=0}^{k-1} 2^{r-1} \E\bigl[x \,\big|\, \delta_{2^{-(r+1)}} < x \leq \delta_{2^{-r}}\bigr] + \eta 2^{k-1}\E\bigl[x \,\big|\, x\leq \delta_{2^{-k}} \bigr].
$$
\end{restatable}

\begin{proof}{Proof.}
By linearity of expectation $\E\bigl[\OPT_n\bigr] = \sum_{i \in [n]} \E\bigl[\min \{x_1,\dots,x_i\} \bigr]$ where for $i \in [n]$ the random variables $x_1,\dots,x_i$ are drawn independently from $X$. To prove the claim we proceed to express for fixed $i \in [n]$ the expectation $\E[\min\{x_1,\dots,x_i\}]$ in terms of $\E\bigl[ x \,\big|\, x \leq \delta_{2^{-k}}\bigr]$ and $\E[x \,|\, \delta_{2^{-(r+1)}} < x \leq \delta_{2^{-r}}]$ with $r \in \{0,\dots,k-1\}$. To this end, for $i \in [n]$ and $r \in \{0,\dots,k-1\}$ let
\begin{align*}
E_{i,r,=1} &= \Bigl[\bigl| \{x_1,\dots,x_i\} \cap (\delta_{2^{-(r+1)}},\delta_{2^{-r}}]\bigr| = 1 \text{ and } \bigl|\{x_1,\dots,x_i\} \cap (\delta_{2^{-(r+1)}},\delta_1]\bigr| = i\Bigr]
\intertext{be the stochastic event that the minimum of the $i$ draws $x_1\sim X,\dots,x_i\sim X$ is in the interval $(\delta_{2^{-(r+1)}},\delta_{2^{-r}}]$ and none of the other $i-1$ draws is in that interval. Additionally, let}
E_{i,k,=1} &= \Bigl[\bigl| \{x_1,\dots,x_i\} \cap [0,\delta_{2^{-k}}]\bigr| = 1\Bigr].
\intertext{Further, for $r \in \{0,\dots,k-1\}$, let}
E_{i,r,>1} &= \Bigl[\bigl| \{x_1,\dots,x_i\} \cap (\delta_{2^{-(r+1)}},\delta_{2^{-r}}]\bigr| > 1 \text{ and } \bigl|\{x_1,\dots,x_i\} \cap (\delta_{2^{-(r+1)}},\delta_1]\bigr| = i\Bigr]
\intertext{be the stochastic event that the minimum of the $i$ draws $x_1\sim X,\dots,x_i\sim X$ is in the interval $(\delta_{2^{-(r+1)}},\delta_{2^{-r}}]$ and at least one of the other $i-1$ draws is in that interval. Similarly, let}
E_{i,k,>1} &= \Bigl[\bigl| \{x_1,\dots,x_i\} \cap [0,\delta_{2^{-k}}]\bigr| > 1\Bigr].
\end{align*}
For fixed $i$ the events $E_{i,r,=1}$ and $E_{i,r,>1}$ for $r \in \{0,\dots,k\}$ are clearly disjoint. Since~$\sum_{r=0}^{k}(\P[E_{i,r,=1}]+\P[E_{i,r,>1}])=1$, by the law of total expectation, we have
\begin{align*}
\E\bigl[\min\{x_1,\dots,x_i\}\bigr] &= \sum_{r=0}^{k} \Bigl( \E\bigl[\min\{x_1,\dots,x_i\} \,\big|\, E_{i,r,=1}\bigr] \P\bigl[E_{i,r,=1}\bigr]\\
&\quad +  \E\bigl[\min\{x_1,\dots,x_i\} \,\big|\, E_{i,r,>1}\bigr] \P\bigl[E_{i,r,>1}\bigr] \Bigr).
\end{align*}
We observe that $\E\bigl[\min\{x_1,\dots,x_i\} \,\big|\, E_{i,r,=1}\bigr] = \E\bigl[x \,\big|\, \delta_{2^{-(r+1)}} < x \leq \delta_{2^{-r}}\bigr]$ for all $r \in \{0,\dots,k-1\}$ and, similarly, $\E\bigl[ \min\{x_1\dots,x_i\} \, \big|\, E_{i,k,=1}\bigr] = \E\bigl[x \,\big|\, x \leq \delta_{2^{-k}}\bigr]$. In addition, we have $\E\bigl[\min\{x_1,\dots,x_i\} \,\big|\, E_{i,r,>1}\bigr] \geq \delta_{2^{-(r+1)}} \geq \E\bigl[x \,\big|\, \delta_{2^{-(r+2)}} < x \leq \delta_{2^{-(r+1)}}\bigr]$ for all $r \in \{0,\dots,k-1\}$. We then obtain
\begin{align*}
\E\bigl[\min\{x_1,\dots,x_i\}\bigr] &\geq \E\bigl[x \,\big|\, \delta_{2^{-1}} < x \leq \delta_{1}\bigr] \P\bigl[E_{i,0,=1}\bigr]\\[6pt]
&\quad + \sum_{r=1}^{k-1} \E\bigl[x \,\big|\, \delta_{2^{-(r+1)}} \!<\! x \!\leq\! \delta_{2^{-r}}\bigr]	\bigl(\P\bigl[ E_{i,r,=1}\bigr] \!+\! \P\bigl[ E_{i,r-1,>1}\bigr] \bigr)\\[6pt]
&\quad + \E\bigl[ x \,\big|\, x \leq \delta_{2^{-k}} \bigr] (\P\bigl[ E_{i,k,=1} \bigr] + \P\bigl[ E_{i,k-1,>1} \bigr]),
\end{align*}
and hence
\begin{align*}
\E\bigl[\OPT_n\bigr] &\geq \E\bigl[x \,\big|\, \delta_{2^{-1}} < x \leq \delta_{1}\bigr] \sum_{i=1}^n \P\bigl[E_{i,0,=1}\bigr]\\[6pt]
&\quad + \sum_{r=1}^{k-1} \Biggl( \E\bigl[x \,\big|\, \delta_{2^{-(r+1)}} < x \leq \delta_{2^{-r}}\bigr]	\sum_{i=1}^n \Bigl(  \P\bigl[ E_{i,r,=1}\bigr] + \P\bigl[ E_{i,r-1,>1}\bigr] \Bigr) \Biggr) \\[6pt]
&\quad + \E\bigl[ x \,\big|\, x \leq \delta_{2^{-k}} \bigr] \sum_{i=1}^n \Bigl(\P\bigl[ E_{i,k,=1} \bigr] + \P\bigl[ E_{i,k-1,>1} \bigr] \Bigr).
\end{align*}
%
The probability that a single draw falls in the range~$(\delta_{2^{-(r+1)}}, \delta_{2^{-r}}]$ and $i-1$ draws are larger than $\delta_{2^{-r}}$ is $2^{-(r+1)}(1-2^{-r})^{i-1}$. 
Since there are~$i$ possibilities which of the draws falls in this range, we have
\begin{align*}
\P\bigl[ E_{i,r,=1}\bigr] &=
\begin{cases}
0 & \text{ if $r=0$ and $i>1$}\\
1/2 & \text{ if  $r=0$ and $i=1$ }\\
i 2^{-(r+1)}(1-2^{-r})^{i-1} & \text{ if $r \in \{1,\dots,k-1\}$ }\\
i2^{-r}(1-2^{-r})^{i-1} & \text{ if $r=k$. }
\end{cases}
\intertext{Similarly, for $r \in \{1,\dots,k\}$, we have}
\P \bigl[ E_{i,r-1,>1} \bigr]
&= (1-2^{-r})^i - (1-2^{-(r-1)})^i - \P\bigl[E_{i,r-1,=1}\bigr]\\[6pt]
&= (1-2^{-r})^i - (1-2^{-(r-1)})^i - i2^{-r}(1-2^{-(r-1)})^{i-1}.
\end{align*}
We then obtain
\begin{align*}
\E[\OPT_n\bigr] &\geq \E\bigl[x \,\big|\, \delta_{2^{-1}} < x \leq \delta_{1} \bigr] \cdot 2^{-1} \\[6pt]
&\quad + \sum_{r =1}^{k-1} \E\bigl[x \,\big|\, \delta_{2^{-(r+1)}} < x \leq \delta_{2^{-r}}\bigr]\cdot \alpha(r,n)\\[6pt]
&\quad + \E\bigl[ x \, \big|\, x \leq \delta_{2^{-k}} \bigr] \cdot \alpha(k,n)
\end{align*}
with
\begin{align}
 \alpha(r,n) &= \sum_{i =1}^n \Bigl( \P\bigl[ E_{i,r,=1} \bigr] + \P\bigl[ E_{i,r-1,>1} \bigr] \Bigr).
 \end{align}
It remains to show that $\alpha(r,n) \geq 2^{r-1}$ for $1\leq r <k$ and $\alpha(k,n) \geq \eta 2^{k-1}$. We have
\begin{align}
 \alpha(k,n) &= \sum_{i =1}^n \Bigl( \P\bigl[ E_{i,r,=1} \bigr] + \P\bigl[ E_{i,r-1,>1} \bigr] \Bigr) \label{eq:alpha}
\intertext{which gives} \notag
 \alpha(r,n) &= \sum_{i=1}^n i 2^{-r}(1\!-\!2^{-r})^{i\!-\!1} + (1\!-\!2^{-r})^i - (1\!-\!2^{-(r-1)})^i - i2^{-r}(1\!-\!2^{-(r-1)})^{i\!-\!1}
\intertext{if $r=k$, and}
 \alpha(r,n) &= \sum_{i=1}^n i 2^{-(r+1)}(1\!-\!2^{-r})^{i\!-\!1} + (1\!-\!2^{-r})^i - (1\!-\!2^{-(r-1)})^i - i2^{-r}(1\!-\!2^{-(r-1)})^{i\!-\!1}\notag
\end{align}
for $r\neq k$.

To prove the lemma, we proceed to show that $\inf_{n \in \mathbb{N}} \inf_{r \in \{0,\dots,\lceil \log n \rceil -3\}} \frac{\alpha(r,n)}{2^{r-1}} \geq 1$ and $\inf_{n \in \mathbb{N}}  \frac{\alpha(\lceil \log n\rceil -2,n)}{2^{\lceil \log n\rceil -3}} \geq \eta$. 
Differentiating the well-known formula for the geometric sum $\sum_{i=1}^n a^i = \frac{a-a^{n+1}}{1-a}$, we obtain $\sum_{i=1}^n i a^{i-1} = \frac{na^{n+1} - (n+1)a^n + 1}{(1-a)^2}$. We use both formulas to simplify all partial sums. For a binary event~$E$, we denote by $\chi_E$ the indicator variable for event $E$, i.e., $\chi_E = 1$ if $E$ is true, and $\chi_E = 0$, otherwise. For $r \in \{1,\dots,k\}$, we then obtain
\begin{align*}
\alpha(r,n) &= (1 + \chi_{r=k})2^{r-1} \bigl[n(1-2^{-r})^{n+1} - (n+1) (1-2^{-r})^n +1\bigr]\\
&\quad + 2^r\bigl[1-2^{-r} - (1-2^{-r})^{n+1}\bigr]\\
& \quad - 2^{r-1}\bigl[1-2^{-(r-1)}\!- (1-2^{-(r-1)})^{n+1}\bigr]\\
&\quad - 2^{r-2}\bigl[n(1-2^{-(r-1)})^{n+1} - (n+1)(1-2^{-(r-1)})^n +1\bigr]\\[6pt]
&= 2^{r-1} \Biggl[ (1+\chi_{r=k})\Bigl(n(1-2^{-r})^{n+1} - (n+1) (1-2^{-r})^n + 1\Bigr)\\
&\quad + 2 - 2^{-(r-1)} - 2(1-2^{-r})^{n+1}\\
&\quad - 1 + 2^{-(r-1)} + (1-2^{-(r-1)})^{n+1}\\
&\quad - \frac{n}{2}(1-2^{-(r-1)})^{n+1} + \frac{n+1}{2}(1-2^{-(r-1)})^n - \frac{1}{2}\Biggr]\\[6pt]
&= 2^{r-1}\Biggl[\frac{3}{2} +\chi_{r=k}\\
&\quad + (1-2^{-r})^n \biggl(n(1+\chi_{r=k})(1-2^{-r}) - (n+1)(1+\chi_{r=k})- 2(1-2^{-r})\biggr)\\
&\quad + (1-2^{-(r-1)})^n \biggl( (1-2^{-(r-1)}) - \frac{n}{2}(1-2^{-(r-1)})+ \frac{n+1}{2}\biggr)\Biggr]\\[6pt]
&= 2^{r-1} \biggl[\frac{3}{2} + \chi_{r=k} + (1-2^{-r})^n(2^{-(r-1)} - n2^{-r}(1+ \chi_{r=k}) - 3 - \chi_{r=k})\\[6pt]
&\quad + (1-2^{-(r-1)})^n \biggl( \frac{3}{2} +n2^{-r}-2^{-(r-1)} \biggr) \biggr]\\[6pt]
&= 2^{r-1} \biggl[ \frac{3}{2} + \chi_{r=k} +\biggl((1-2^{-r})^n - (1-2^{-(r-1)})^n\biggr) \biggl(2^{-(r-1)} - n2^{-r} - \frac{3}{2}\biggr)\\[6pt]
&\quad - \biggl(\frac{3}{2}+\chi_{r=k}(1+n2^{-r})\biggr)(1-2^{-r})^n \biggr].
\end{align*}
As the probabilities are non-negative, $\alpha(r,n) =  \sum_{i=1}^n (\P[E_{i,r,=1}] + \P[E_{i,r-1,>1}])$ is  non-decreasing in $n$ for all $r \in \{1,\dots,k\}$. We proceed to show that $\alpha(r,n) \geq 2^{r-1}$ for all $r \in \{0,\dots,k-1\} = \{0,\dots, \lceil \log n \rceil -3\}$. Since~$r,n$ are integral, $r\leq \lceil \log n\rceil - 3$ implies~$n\geq 2^{r+2} + 1$.
Using monotonicity of~$\alpha(r,n)$ and substituting~$t := 2^r$, we have
\begin{align*}
&\inf\nolimits_{n \in \mathbb{N}} \inf\nolimits_{r \in \{1,\dots,\lceil \log n \rceil -3\}} \frac{\alpha(r,n)}{2^{r-1}}\\[6pt]
&=\inf\nolimits_{r\in\mathbb{N}} \inf\nolimits_{n\in\{2^{r+2} + 1, \dots\}} \frac{\alpha(r,n)}{2^{r-1}}\\[6pt]
&= \inf\nolimits_{r \in \mathbb{N}} \frac{\alpha(r, 2^{r+2}+1)}{2^{r-1}}\\[6pt]
&= \inf\nolimits_{r \in \mathbb{N}} \Biggl\{ \frac{3}{2}  \!+\! \biggl((1\!-\!2^{-r})^{2^{r+2}\!+\!1} \!-\! (1\!-\!2^{-(r-1)})^{2^{r+2}\!+\!1}\biggr) \!\biggl(2^{-(r\!-\!1)} \!-\! (2^{r+2}+1)2^{-r} \!-\! \frac{3}{2} \biggr)\\
&\quad\quad\quad\quad\quad-\frac{3}{2}(1-2^{-r})^{2^{r+2}+1} \Biggr\}\\[6pt]
&\geq \inf\nolimits_{t \in \mathbb{N}} \Biggl\{ \frac{3}{2}  + \left(\biggl(1\!-\!\frac{1}{t}\biggr)^{\!4t\!+\!1} - \biggl(1\!-\!\frac{2}{t}\biggr)^{\!\!4t\!+\!1}\right) \!\Biggl(\frac{2}{t} \!-\! 4 \!-\! \frac{1}{t} \!-\! \frac{3}{2} \Biggr) -\frac{3}{2}\Bigl(1\!-\!\frac{1}{t}\Bigr)^{\!4t\!+\!1} \Biggr\}\\[6pt]
&= \inf\nolimits_{t \in \mathbb{N}} \Biggl\{ \frac{3}{2}  + \Biggl(\biggl(1-\frac{1}{t}\biggr)^{\!4t+1} - \biggl(1-\frac{2}{t}\biggr)^{\!4t+1}\Biggr) \Biggl(\frac{1}{t} - \frac{11}{2} \Biggr) -\frac{3}{2}\Bigl(1-\frac{1}{t}\Bigr)^{\!\!4t+1} \Biggr\}.
\intertext{The first order Taylor approximation of the function $f(x) = x^{4t+1}$ at $x = 1-\nicefrac{1}{t}$ gives $f(1-\nicefrac{2}{t})=(1 - \nicefrac{2}{t})^{4t+1} = (1- \nicefrac{1}{t})^{4t+1} - \frac{4t+1}{t} (1- \nicefrac{1}{t})^{4t} + R_2$, with $R_2 \geq 0$ as $f$ is convex. This implies} 
&\inf\nolimits_{n \in \mathbb{N}} \inf\nolimits_{r \in \{1,\dots,\lceil \log n \rceil -2\}} \frac{\alpha(r,n)}{2^{r-1}}\\[6pt]
&\geq \inf\nolimits_{t \in \mathbb{N}} \Biggl\{ \frac{3}{2} -  \frac{4t+1}{t}\biggl(\frac{11}{2} - \frac{1}{t}\biggr)\biggl(1-\frac{1}{t}\biggr)^{\!\!4t}  -\frac{3}{2}\biggl(1-\frac{1}{t}\biggr)\biggl(1-\frac{1}{t}\biggr)^{\!\!4t} \Biggr\}\\[6pt]
&= \inf\nolimits_{t \in \mathbb{N}} \Biggl\{ \frac{3}{2} -  \biggl(\frac{47}{2} - \frac{1}{t^2}\biggr)\biggl(1-\frac{1}{t}\biggr)^{\!\!4t} \Biggr\} \\[6pt]
&\geq \inf\nolimits_{t \in \mathbb{N}} \Biggl\{ \frac{3}{2} -  \frac{47}{2}\biggl(1-\frac{1}{t}\biggr)^{\!\!4t} \Biggr\}.
\end{align*}
As the latter expression is decreasing in $t$, we have $$\inf\nolimits_{n \in \mathbb{N}} \inf\nolimits_{r \in \{1,\dots,\lceil \log n \rceil -3\}} \frac{\alpha(r,n)}{2^{r-1}} \geq \lim_{t \to \infty} \left\{\frac{3}{2} - \frac{47}{2}(1-\frac{1}{t})^{4t}\right\} = \frac{3}{2} - \frac{47}{2\e^4} \approx 1.069 > 1.$$
It remains to show that $\frac{\alpha(k,n)}{2^{k-1}} \geq \eta$. For $r=k = \lceil \log n \rceil -2$, we have
\begin{align*}
\alpha(k,n) &= 2^{k-1} \biggl[\frac{5}{2} +\biggl((1-2^{-k})^n - (1-2^{-(k-1)})^n\biggr) \biggl(2^{-(k-1)} - n2^{-k} - \frac{3}{2}\biggr)\\
&\quad - \biggl(\frac{5}{2}+n2^{-k}\biggr)(1-2^{-k})^n \biggr].
\end{align*}
Again, as $\alpha(r,n)$ is non-decreasing in $n$, this value is minimal for $n=2^{k+1}+1$. Substituting $t=2^k$, we obtain
\begin{align*}
&\inf_{n \in \mathbb{N}} \frac{\alpha(\lceil \log n \rceil -2,n)}{2^{\lceil \log n \rceil -3}}\\[6pt]
&= \inf_{k \in \mathbb{N}} \frac{\alpha(k,2^{k+1}+1)}{2^{k-1}}\\[6pt]
&= \inf_{k \in \mathbb{N}} \Biggl\{ \frac{5}{2} +\biggl((1\!-\!2^{-k})^{2^{k+1}\!+\!1} - (1\!-\!2^{-(k-1)})^{2^{k+1}\!+\!1}\biggr) \biggl(2^{-(k\!-\!1)} - (2^{k+1}\!+\!1)2^{-k} - \frac{3}{2}\biggr) \\[6pt]
&\quad - \biggl(\frac{5}{2}+(2^{k+1}+1)2^{-k}\biggr)(1-2^{-k})^{2^{k+1}+1} \Biggr\}\\[6pt]
&\geq \inf_{\substack{t \in \mathbb{N}\\t\geq 2}} \Biggl\{ \frac{5}{2} +\Biggl( \!\biggl(1\!-\!\frac{1}{t}\biggr)^{\!\!2t\!+\!1} \!\!- \biggl(1\!-\!\frac{2}{t}\biggr)^{\!\!2t\!+\!1}\Biggr) \!\Biggl(\frac{2}{t} \!-\! 2 -\frac{1}{t} \!-\! \frac{3}{2}\Biggr) \!-\! \biggl(\frac{5}{2} \!+\! 2 \!+\! \frac{1}{t}\biggr)\biggl(1\!-\!\frac{1}{t}\biggr)^{\!\!2t\!+\!1} \Biggr\}\\[6pt]
&= \inf_{\substack{t \in \mathbb{N}\\t\geq 2}} \Biggl\{ \frac{5}{2} +\Biggl( \biggl(1-\frac{1}{t}\biggr)^{\!\!2t+1} - \biggl(1-\frac{2}{t}\biggr)^{\!\!2t+1}\Biggr) \Biggl(\frac{1}{t} - \frac{7}{2}\Biggr) - \biggl(\frac{9}{2} + \frac{1}{t}\biggr)\biggl(1-\frac{1}{t}\biggr)^{\!\!2t+1} \Biggr\}.
\end{align*}
By second-order Taylor approximation of the function $f(x) = x^{2t+1}$ at $x = (1 - \nicefrac{1}{t})$, we obtain
\begin{align*}
f(1-\frac{2}{t})=\biggl(1 - \frac{2}{t}\biggr)^{\!\!2t+1} &= \biggl(1-\frac{1}{t}\biggr)^{\!\!2t+1} - \frac{2t+1}{t}\biggl(1-\frac{1}{t}\biggr)^{\!\!2t} + \frac{2t(2t+1)}{2t^2}\biggl(1-\frac{1}{t}\biggr)^{\!\!2t-1}\\[6pt]
 &\quad - \frac{2t(2t+1)(2t-1)}{6t^3}\biggl(1- \frac{1}{t}\biggr)^{\!\!2t-2} + R_4,
\end{align*}
where the remainder is $R_4 \geq 0$, as the fourth derivative is non-negative. (This can easily be seen when expressing the remainder in Lagrange form.)
We then obtain
\begin{align*}
&\inf_{n \in \mathbb{N}} \frac{\alpha(\lceil \log n \rceil -2,n)}{2^{\lceil \log n \rceil -3}}\\[6pt]
&\geq \inf_{\substack{t \in \mathbb{N}\\t\geq 2}} \Biggl\{ \frac{5}{2} +\Biggl[ \frac{2t+1}{t}\biggl(1-\frac{1}{t}\biggr)^{\!\!2t} - \frac{2t+1}{t}\biggl(1-\frac{1}{t}\biggr)^{\!\!2t-1}\\
&\quad\quad\quad\quad\quad + \frac{(2t+1)(2t-1)}{3t^2}\biggl(1-\frac{1}{t}\biggr)^{\!\!2t-2} \Biggr] \Biggl[\frac{1}{t} - \frac{7}{2}\Biggr] - \biggl(\frac{9}{2} + \frac{1}{t}\biggr)\biggl(1-\frac{1}{t}\biggr)^{\!\!2t+1} \Biggr\}\\[6pt]
&= \inf_{\substack{t \in \mathbb{N}\\t\geq 2}} \Biggl\{ \frac{5}{2} +\biggl(1-\frac{1}{t}\biggr)^{\!\!2t} \Biggl[ \biggl(\frac{2t+1}{t} - \frac{2t+1}{t-1} + \frac{(2t+1)(2t-1)}{3(t-1)^2}\biggr) \biggl(\frac{1}{t} - \frac{7}{2}\biggr)\\
&\quad\quad\quad\quad\quad  - \biggl(\frac{9}{2} + \frac{1}{t}\biggr)\biggl(1-\frac{1}{t}\biggr)\Biggr] \Biggr\}\\
&= \inf_{\substack{t \in \mathbb{N}\\t\geq 2}} \Biggl\{ \frac{5}{2}  - \frac{55t^4 - 125t^3 + 89t^2 + 8t - 12}{6t^2(t-1)^2}\biggl(1-\frac{1}{t}\biggr)^{\!\!2t}  \Biggr\}.
\end{align*}
It is straightforward to check that $(1-\nicefrac{1}{t})^{2t}$ and  $\frac{55t^4 - 125t^3 + 89t^2 + 8t - 12}{6t^2(t-1)^2}$ are increasing in $t$. This implies
\begin{align*}
\inf_{k \in \mathbb{N}} \frac{\alpha(k,2^{k+1}+1)}{2^{k-1}}
&= \lim_{t \to \infty} \Biggl\{ \frac{5}{2}  - \frac{55t^4 - 125t^3 + 89t^2 + 8t - 12}{6t^2(t-1)^2}\biggl(1-\frac{1}{t}\biggr)^{\!\!2t}  \Biggr\}\\[6pt]
&= \frac{5}{2} - \frac{55}{6\e^2} \approx 1.259,
\end{align*}
which finishes the proof.
\hfill\Halmos\end{proof}

Combining Lemmas~\ref{lem:uniform_transitions} and~\ref{lem:arb_dist_offline_opt}, we obtain
the main result of this section.

\begin{restatable}{theorem}{thmarbitrary}
  Algorithm~\ref{alg:arb_dist} is $\ratioarbitrary$-competitive for arbitrary distributions.
\end{restatable}

\begin{proof}{Proof.}
For $n \leq 4$, Algorithm~\ref{alg:arb_dist} hires the first applicant for the whole time which gives  $4$-approximation. For the following arguments, assume that $n \geq 5$, and let $k = \lceil \log n \rceil - 2 \geq 1$.

Algorithm~\ref{alg:arb_dist} hires an applicant, whenever the Markov chain transitions from a state $B_j$ to $A_{j+1}$ and hires the final applicant when it reaches state $B_k$. 
By Lemma~\ref{lem:uniform_transitions} for each $j$, the expected number of transitions from state $B_j$ to $A_{j+1}$ is at most $\smash{\frac{p}{3p-1}}$ where $p = 1 - \nicefrac{1}{\e}$.
Each applicant who is hired while transitioning from $B_j$ to $A_{j+1}$ is hired for $\smash{2^{j+2}}$ time units, and its expected cost value is
$\Exp{ x \mid \delta_{2^{-(j+1)}} \leq  x \leq\delta_{2^{-j}}}$. The final applicant hired when state $B_k$ is reached is hired for at most $n$ time units and has expected cost with value $\E[ x \,|\, x \leq \delta_{2^{-k}}]$.
%

Since the number of visits to a state and the cost for hiring an applicant in the state are stochastically independent, we obtain
\begin{align}
\Exp{\ALG_n} &\leq \frac{p}{3p-1} \sum_{j=0}^{k-1} \Bigl(2^{j+2} \E\bigl[ x \,\big|\, \delta_{2^{-(j+1)}} < x \leq \delta_{2^{-j}}\bigr] \Bigr) + n \E[ x\,|\, x \leq \delta_{2^{-k}}] \notag\\[6pt]
&= \frac{1-\nicefrac{1}{\e}}{2-\nicefrac{3}{\e}} \sum_{j=0}^{k-1} \Bigl(2^{j+2} \E\bigl[ x \,\big|\, \delta_{2^{-(j+1)}} < x \leq \delta_{2^{-j}}\bigr] \Bigr) + n \E[ x\,|\, x \leq \delta_{2^{-k}}] \notag\\[6pt]
&\leq \frac{8e-8}{2e-3} \E[\OPT_n] + n\biggl(1 - \frac{\e-1}{2\e-3}\eta\biggr)\E\bigl[ x \,\big|\, x \leq \delta_{2^{-k}}\bigr] \label{eq:alg_n},
\end{align}
where we used Lemma~\ref{lem:arb_dist_offline_opt} and where $\eta = \frac{5}{2} - \frac{55}{6\e^2}$. Further, recall that $\E[\OPT_n] = \sum_{i=1}^n \E[\min \{x_1,\dots,x_i\}]$. For $i \in [n]$, we have
\begin{align}
\E\bigl[\min \{x_1,\dots,x_i\}\bigr] &\geq \E\bigl[x \,\big|\, x \leq \delta_{2^{-k}}\bigr] \P\bigl[\left|\{x_1,\dots,x_i\} \cap [0,\delta_{2^{-k}}]\right| \leq 1\bigr] \notag \\[6pt]
&\geq \E\bigl[x \,\big|\, x \leq \delta_{2^{-k}}\bigr] \P\bigl[\left|\{x_1,\dots,x_i\} \cap [0,\delta_{4/n}]\right| \leq 1\bigr] \notag \\[6pt]
&= \E\bigl[x \,\big|\, x \leq \delta_{2^{-k}}\bigr] \bigl( \P\left[\left| \{x_1,\dots,x_i\} \cap [0,\delta_{4/n}]\right| = 0 \right]\notag\\[6pt] 
&\phantom{=\E\bigl[x \,\big|\, x \leq \delta_{2^{-k}}\bigr] \bigl(} + \P\left[\left| \{x_1,\dots,x_i\} \cap [0,\delta_{4/n}]\right| = 1 \right] \bigr)\notag\\[6pt]
&= \E\bigl[x \,\big|\, x \leq \delta_{2^{-k}}\bigr] \Biggl( \biggl(1- \frac{4}{n}\biggr)^{\!\!i} + \frac{4i}{n}\biggl(1-\frac{4}{n}\biggr)^{\!\!i-1}\Biggr), \notag
\intertext{which implies (for $n\geq5$)}
\E\bigl[\OPT_n] 
&\geq \E\bigl[x \,\big|\, x \leq \delta_{2^{-k}}\bigr] \sum_{i=1}^n \Biggl( \biggl(1-\frac{4}{n}\biggr)^{\!\!i} + \frac{4i}{n}\biggl(1-\frac{4}{n}\biggr)^{\!\!i-1} \Biggr) \notag\\[6pt]
&= \E\bigl[x \,\big|\, x \leq \delta_{2^{-k}}\bigr] \Biggl(\frac{n}{2}\Biggl(1 - 3\biggl(1-\frac{4}{n}\biggr)^{\!\!n} \Biggr) + \biggl(1-\frac{4}{n}\biggr)^{\!\!n} -1 \Biggr) \notag\\[6pt]
&\geq \E\bigl[x \,\big|\, x \leq \delta_{2^{-k}}\bigr] \Biggl(\frac{n}{2}\biggl(1-\frac{3}{\e^4}\biggr)+ \biggl(1-\frac{4}{n}\biggr)^{\!\!n} - 1 \Biggr).\label{eq:final_probability}
\end{align}
Combining \eqref{eq:final_probability} with \eqref{eq:alg_n} and using $n \geq 5$, we obtain
\begin{align*}
\E[\ALG_n] &\leq \frac{8\e-8}{2\e-3}\E[\OPT_n] +  \frac{1 - \frac{\e-1}{2\e-3}(\frac{5}{2} - \frac{55}{6\e^2})}{\frac{1}{2}\bigl(1-\frac{3}{\e^4}\bigr)+ \left(1-\frac{4}{n}\right)^{\!\!n} -\frac{1}{n}}\E[\OPT_n]\\[6pt]
&\leq \biggl( \frac{8\e-8}{2\e-3} + \frac{1 - \frac{\e-1}{2\e-3}(\frac{5}{2} - \frac{55}{6\e^2})}{\frac{1}{2}\bigl(1-\frac{3}{\e^4}\bigr)+ \left(1-\frac{4}{5}\right)^{\!\!5} -\frac{1}{5}}\biggr)\E[\OPT_n] \leq \ratioarbitrary \cdot\E[\OPT_n]
\end{align*}
as claimed.
\hfill\Halmos
\end{proof}

%
%

\section{Unknown Distributions}\label{sec:unknown_distribution}

In this section, we again consider an arbitrary distribution $X$ with distribution
function $F$. In contrast to before, we assume that
$X$ is unknown to us. In particular, we do not have access
to the quantiles of $X$. We first give a bound for the
cost of the offline optimum that does not rely on quantiles.
In the following, we let $\Ex{x} := \Exx{x\sim X}{x}$.

\begin{restatable}{lemma}{lemundistofflineopt}
For arbitrary distributions $X$, \label{lem:un_dist_offline_opt2}
$\Exp{\OPT_n}\geq\Exp x+\sum_{i=1}^{\lfloor\log n\rfloor}2^{i-1}\int_{0}^{\infty}(1-F(x))^{2^{i}}\intd x$.
\end{restatable}

\begin{proof}{Proof.}
Since the left hand side of the inequality to prove is increasing in $n$ while the right hand side only increases when $n$ is a power of $2$, we may assume without loss of generality that $n$ is a power of $2$. By Proposition~\ref{pro:opt}, we have
$\Exp{\OPT_n}=\sum\nolimits _{i\in[n]}\int_{0}^{\infty}(1-F(x))^{i}\intd x$.
Using that $(1-F(x))^{i}$ is decreasing with $i$, we split the sum into the ranges $(\nicefrac{n}{2},n],(\nicefrac{n}{4},\nicefrac{n}{2}],(\nicefrac{n}{8},\nicefrac{n}{4}],\dots$
and bound each part by the last term in the corresponding range, i.e., 
\begin{align*}
\Exp{\OPT_n} & = \Exp x+\sum\nolimits _{i=2}^{n}\int_{0}^{\infty}(1-F(x))^{i}\intd x\\[6pt]
 & \geq \Exp x+\frac{n}{2}\int_{0}^{\infty}(1-F(x))^{n}\intd x+\frac{n}{4}\int_{0}^{\infty}(1-F(x))^{n/2}\intd x  +\dots\\
 &\quad\dots+\int_{0}^{\infty}(1-F(x))^{2}\intd x\\[6pt]
 & = \Exp x+\sum\nolimits _{i=1}^{\log n}\frac{n}{2^{i}}\int_{0}^{\infty}(1-F(x))^{n/2^{i-1}}\intd x\\[6pt]
 & = \Exp x+\sum\nolimits _{i=1}^{\log n}2^{\log n-i}\int_{0}^{\infty}(1-F(x))^{2^{\log n-i+1}}\intd x\\[6pt]
 & = \Exp x+\sum\nolimits _{i=0}^{\log n-1}2^{i}\int_{0}^{\infty}(1-F(x))^{2^{i+1}}\intd x\\[6pt]
 & = \Exp x+\sum\nolimits _{i=1}^{\log n}2^{i-1}\int_{0}^{\infty}(1-F(x))^{2^{i}}\intd x,
\end{align*}
as claimed.\hfill \Halmos
\end{proof}

We now describe our algorithm for unknown distributions (cf.~Algorithm~\ref{alg:unknown_dist}).
Without knowledge of the quantiles of $X$, we have no good
way to directly adjust the cost threshold $\ct$. 
Instead, for some integral value $\lambda > 1$ to be fixed later, we devote a $\nicefrac{1}{\lambda+1}$ fraction of the time spent in each state~$j$ to sample $X$
in order to estimate a suitable value for~$\ct$ and then wait for an appropriate candidate to appear. Specifically, in state~$j$ we sample for $2^j-1$ time units and then observe the applicants for another $\lambda (2^j-1)$ time units. Thus, the maximum number of time units spent in state $j$ is $\bar{t}_{j}=(1+\lambda)(2^{j}-1)$. When observing the applicants we hire any candidate whose cost does not exceed the minimum cost while sampling. The hiring time is $t_{j}=(1+\lambda)2^{j+2}$ time units. Since
\begin{align*}
\sum_{i=0}^{j+1} \bar{t}_{i} = (1+\lambda)\sum_{i=0}^{j+1} (2^{i}-1) = (1+\lambda)(2^{j+2}-j-3) \leq t_j
\end{align*}
we are guaranteed to hire a new applicant (or terminate the algorithm) during
the hiring time.  

\begin{algorithm}[tb]
\caption{A $\ratiounknown$-competitive algorithm for unknown distributions.}
\label{alg:unknown_dist}
 $\ct \set \infty$ \tcp*[r]{threshold cost}
 $t_\mathrm{sample} \set 0$ \tcp*[r]{remaining time until threshold is fixed}
 $t_\mathrm{wait} \set 1$ \tcp*[r]{remaining time once threshold is fixed}
 $j \set 0$ \tcp*[r]{state of the algorithm}
\For{$i \set 1, \dots, n$  }{ 
  \If{$t_\mathrm{sample} > 0$}{
    $\ct \set \min\{\ct, x_i\}$\;
    $t_\mathrm{sample} \set t_\mathrm{sample} - 1$\;
  }
  \ElseIf{$t_\mathrm{wait} > 0$}{
    $t_\mathrm{wait} \set t_\mathrm{wait} - 1$\;
    \If{$x_i \le \ct$}{
      hire applicant $i$ for $(1+\lambda)2^{j+2}$ time steps\;
      \If{$i + (1+\lambda)2^{j+2} > n$}{
       {\bf stop}\;
      }
      $j \set j + 1$; $\ct \set \infty$; $t_\mathrm{sample} \set 2^j - 1$; $t_\mathrm{wait} \set \lambda t_\mathrm{sample}$\;
    }
  }
  \Else{
    $j \set j - 1$; $\ct \set \infty$; $t_\mathrm{sample} \set 2^j - 1$; $t_\mathrm{wait} \set \lambda t_\mathrm{sample}$\;
  }
}
\end{algorithm}

The maximum value of $j$ that can be reached during the execution
of the algorithm is bounded by the fact that $(1+\lambda)2^{j+2}\leq n$, i.e.,
$j \leq \lceil \log\frac{n}{1+\lambda} \rceil -2$.

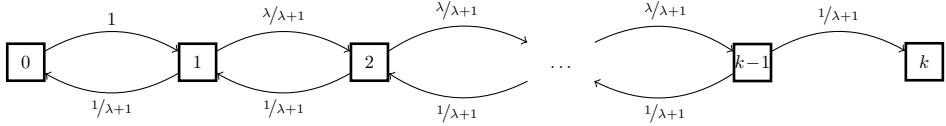
\begin{figure}
\centering
\begin{tikzpicture}[scale=0.3, every node/.style={scale=0.7}]
\matrix (m) [matrix of nodes, row sep=3em, column sep=5em,
	nodes={draw, rectangle, line width=1, text centered, minimum size=2em}]
    {\\  $0$ & $1$ & $2$  & |[draw=none]|$\phantom{1k}\dots\phantom{1k}$ & $\!\!k\!-\!1\!\!$  & $k$ \\ };

\draw[->] (m-2-1) edge [bend left] node [midway,above] {$1$} (m-2-2);
\draw[->] (m-2-2) edge [bend left] node [midway,below] {$\nicefrac{1}{\lambda+1}$} (m-2-1);

\draw[->] (m-2-2) edge [bend left] node [midway,above] {$\nicefrac{\lambda}{\lambda+1}$} (m-2-3);
\draw[->] (m-2-3) edge [bend left] node [midway,below] {$\nicefrac{1}{\lambda+1}$} (m-2-2);

\draw[->] (m-2-3) edge [bend left] node [midway,above] {$\nicefrac{\lambda}{\lambda+1}$} (m-2-4);
\draw[->] (m-2-4) edge [bend left] node [midway,below] {$\nicefrac{1}{\lambda+1}$} (m-2-3);

\draw[->] (m-2-4) edge [bend left] node [midway,above] {$\nicefrac{\lambda}{\lambda+1}$} (m-2-5);
\draw[->] (m-2-5) edge [bend left] node [midway,below] {$\nicefrac{1}{\lambda+1}$} (m-2-4);

\draw[->] (m-2-5) edge [bend left] node [midway,above] {$\nicefrac{1}{\lambda+1}$} (m-2-6);


\foreach \x [count=\y] in {2,3}{
}

\end{tikzpicture}
\caption{Markov chain $M(p,k)$ with $p = \lambda/(\lambda+1)$. \label{fig:markov_unknown}}
\end{figure}

Again, we introduce a Markov chain that has one state for each possible value of $j$ and an absorbing state $k$, see Figure~\ref{fig:markov_unknown}. The probability that we do not hire an applicant in state $j$ with $0<j<k$ equals the probability
that the smallest cost observed while sampling is lower than the smallest
cost observed while waiting. Since $t_{\mathrm{wait}}=\lambda t_{\mathrm{sample}}$,
we have a hiring probability of $p=\nicefrac{\lambda}{\lambda+1}$.
With this probability, the Markov chain transitions to state~$j+1$,
otherwise to state $j-1$.

As the Markov chain already has homogenous transition probabilities equal to $p = \lambda/(\lambda+1)$, Lemma~\ref{lem:markov_a_visit} directly implies the following result.

\begin{restatable}{lemma}{lemarbdistexpectedvisits}
The expected number of visits to each state $j$ of the Markov chain is at most~$\frac{1}{2p-1}$.\label{lem:arb_dist_expected_visits}
\end{restatable}



Combining Lemma~\ref{lem:un_dist_offline_opt2} and Lemma~\ref{lem:arb_dist_expected_visits} yields the main result of this section.

\begin{restatable}{theorem}{thmunknown}
For $\lambda=3$,  Algorithm~\ref{alg:unknown_dist} is strictly $\ratiounknown$-competitive for unknown distributions.
\end{restatable}

\begin{proof}{Proof.}
Using Lemma~\ref{lem:arb_dist_expected_visits} with $p= \frac{\lambda}{\lambda+1}$, we conclude that the algorithm visits
each state at most $\frac{1}{2p-1} = \frac{\lambda+1}{\lambda-1}$ times in expectation.
In state~$j$ with $0<j<k$ with probability $p= \frac{\lambda}{\lambda+1}$ an applicant is
hired for $(1+\lambda)2^{j+2}$ units of time. 
The cost of the applicant is determined by drawing $2^{j}-1$ numbers to determine a minimum $\ct$, and then continuing to draw until we find the first cost smaller than $\ct$.
We can bound the expected cost of the applicant by the expected cost when drawing $2^j$ numbers and taking the minimum, i.e.,
\[
\Exp{x\,|\,x\leq\ct}\leq\mathbf{E}_{x_{i}\sim X}[\min\nolimits _{i\in\{1,\dots,2^{j}\}}\{x_{i}\}]=\int_{0}^{\infty}(1-F(x))^{2^{j}}\intd x.
\]
The algorithm stops at the latest when an applicant is hired in state $k-1 = \lceil \log \frac{n}{\lambda+1} \rceil - 2$ as the applicant is hired for at least $n$ time steps.  
Since the number of visits to a state, the probability of hiring in a state, and the expected cost when hiring
 are independent, we obtain
\[
\Exp{\ALG_n}\leq\frac{\lambda(\lambda+1)}{\lambda-1}\sum_{j=0}^{k}\Bigg(2^{j+2}\int_{0}^{\infty}(1-F(x))^{2^{j}}\intd x\Bigg).
\]
Together with Lemma~\ref{lem:un_dist_offline_opt2} and $k-1\leq\lfloor\log n\rfloor-2$,
this yields
\begin{align*}
\frac{\Exp{\ALG_n}}{\Exp{\OPT_n}} &\leq \frac{4 \Exp x\frac{(\lambda + 1)^2}{\lambda - 1} +\sum_{j=1}^{k-1}\bigl(\frac{\lambda(\lambda + 1)}{\lambda -1}2^{j+2}\int_{0}^{\infty}(1-F(x))^{2^{j}}\intd x \bigr)}{\Exp x+\sum_{j=1}^{\lfloor\log n\rfloor}2^{j-1}\int_{0}^{\infty}(1-F(x))^{2^{j}}\intd x\phantom{)}}\\
&\leq \max\Biggl\{\frac{4 \Exp x\frac{(\lambda + 1)^2}{\lambda - 1}}{\Exp x}, \frac{\sum_{j=1}^{k-1}\bigl(\frac{\lambda(\lambda + 1)}{\lambda -1}2^{j+2}\int_{0}^{\infty}(1-F(x))^{2^{j}}\intd x \bigr)}{\sum_{j=1}^{\lfloor\log n\rfloor}2^{j-1}\int_{0}^{\infty}(1-F(x))^{2^{j}}\intd x\phantom{)}}\Biggr\}\\
&\leq \max\Biggl\{4\frac{(\lambda +1)^2}{\lambda -1}, 8 \frac{\lambda(\lambda +1)}{\lambda -1}\Biggr\} \leq \ratiounknown,
\end{align*}
as claimed.\hfill \Halmos
\end{proof}
\section{Sequential Employment}

\label{sec:sequential_employment}

We now turn our attention to the number of applicants that are
concurrently under employment. We show that there is no constant competitive algorithm for the problem that the covering constraint for the required number of employed candidates is fulfilled with equality in every step.

We can easily adapt the algorithms in the previous sections to be competitive in a setting where not
more than two applicants may be employed during any period of time.
\begin{lemma}\label{lem:only-two}
We can adapt each of the above algorithms to employ not more than two applicants concurrently and ensure them to only lose a factor of at most $2$ in their competitive ratio. 
\end{lemma}
\begin{proof}{Proof.}
We double the hiring times of the algorithms and stay idle during
the first half of the hiring period, i.e., we discard all applicants
encountered during that period. This doubling causes a loss of a factor not larger than $2$. Further, it has the effect that after waiting
for half of the hiring time, effectively, the remaining hiring time
is as before. This in turn implies that the employment period of any
previously hired applicant runs out while staying idle for a new applicant.
This is because the hiring time of a new applicant was defined to
be larger than the remaining hiring time of the previous one, and thus only ever two applicants are employed concurrently. \hfill \Halmos
\end{proof}

Lemma~\ref{lem:only-two} allows us to generalize our algorithms for input sequences of unknown length. Without knowledge of~$n$, we cannot stop our algorithm once an applicant is 
hired for more than the remaining time. 
However, if no more than two applicants are employed concurrently, it is guaranteed that we never employ more than a single 
additional applicant.

\begin{corollary}
Algorithms~\ref{alg:multi_uniform}--\ref{alg:unknown_dist} can be adapted to be competitive even when
	$n$ is not known.
\end{corollary}

The question remains, whether we can stay competitive when only a single
applicant may be employed at a time. We refer to this setting as the
setting of \emph{sequential employment.} 
In the remaining part of this section, we show that the competitive ratio is $\Omega\!\left(\notnicefrac{\sqrt{n}}{\log n}\right)$ for any online algorithm, even when $X=\U[0,1]$. 
Note that the offline
optimum only uses sequential employment.

Let $\E_{n}$ denote the expected cost of the best online algorithm for $n$ applicants under sequential employment.
We give an optimal online algorithm (cf.~Algorithm~\ref{alg:sequential}) based on the values $\E_1, \E_2,\dots, \E_{n-1}$.
Since a single applicant needs to be employed at any time, the only
decision of the algorithm regards the respective hiring times. Interestingly,
our algorithm hires all but the last applicant only for a single unit
of time.

Before we prove this result, we need the following technical lemma.

\begin{restatable}{lemma}{lemsequentialtechnical}
The function $G(\tau):=\Prob{x\geq\tau}(\tau-\Exp{x\,|\,x\geq\tau})$
is non-decreasing.\label{lem:sequential_G}
\end{restatable}

\begin{proof}{Proof.}
We rewrite $G(\ct) = \ct\Prob{x\geq\ct}-\int_{\tau}^{\infty}xf(x)\intd x$ where $f$ is the density of $X$. Then, for $\ct'>\ct$, we have
\begin{align*}
G(\ct')-G(\ct) & = \ct'\Prob{x\geq\ct'}-\ct\Prob{x\geq\ct}+\int_{\ct}^{\ct'}xf(x)\intd x\\[6pt]
 & \geq \ct'\int_{\ct'}^{\infty} f(x)\intd x-\ct\int_{\ct}^{\infty} f(x)\intd x+\ct\int_{\ct}^{\ct'}\,f(x)\intd x\\[6pt]
 & \geq 0,
\end{align*}
which concludes the proof. \hfill \Halmos
\end{proof}

\begin{algorithm}[tb]
\caption{An optimal online algorithm for sequential employment.}
\label{alg:sequential}
\For{$i \set 1, \dots, n$}{ 
  \If{$x_i < \tau_{n-i} = \frac{\E_{n-i}}{n-i}$} {
    hire applicant $i$ for remaining time $n-i+1$\;
    \bf{stop}\;
  }
  \Else {
    hire applicant $i$ for one unit of time\;
  }
}
\end{algorithm}
We are now in position to prove that Algorithm~\ref{alg:sequential} is optimal.

\begin{restatable}{theorem}{thmsequentialopt}
Algorithm~\ref{alg:sequential} is an optimal online algorithm for
sequential employment.
\end{restatable}

\begin{proof}{Proof.}
Let $\ct_{i}:=\lfrac{\E_{i}}{i}$ be the threshold employed by Algorithm~\ref{alg:sequential} when $i$ applicants remain. For technical reasons, let $\tau_0$ be any constant greater than $\tau_1$. We prove the theorem by induction on $n$, additionally showing that $\ct_{n}\leq\ct_{n-1}$. For $n=1$, the algorithm is obviously optimal and $\tau_1 \leq \tau_0$ by definition. Consider the first applicant of cost $x_{1}$. With $\E_{0}:=0$, the expected cost of the optimal online algorithm follows the recursion 
\begin{align}
\min_{t\in\{1,\dots,n\}}\{x_{1}t+\E_{n-t}\}.\label{eq:opt_online}
\end{align}

Consider the case $x_{1}<\tau_{n-1} = \E_{n-1} / (n-1)$. We proceed to show that the minimum
(\ref{eq:opt_online}) is attained for $t=n$. By induction, for all
$t\in\{1,\dots,n-1\}$, we have $x_{1}<\ct_{n-1}\leq\ct_{t}$, and
thus
\[
nx_{1}=tx_{1}+(n-t)x_{1}<tx_{1}+\E_{n-t}.
\]

Now consider the case $x_{1}\geq\ct_{n-1}$. We need to show that
the minimum (\ref{eq:opt_online}) is attained for $t=1$. By induction,
for all $t\in\{2,\dots,n\}$, we have $\ct_{n-1}\leq\ct_{n-t}$, and
thus
\begin{eqnarray*}
tx_{1}+\E_{n-t} & = & x_{1}+(t-1)x_{1}+(n-t)\ct_{n-t}\\
 & \geq & x_{1}+(t-1)\ct_{n-1}+(n-t)\ct_{n-1}\\
 & = & x_{1}+\E_{n-1}.
\end{eqnarray*}

It remains to show $\ct_{n}\leq\ct_{n-1}$. From the above, we have
\[
\E_{n}=n\Prob{x<\ct_{n-1}} \Exp{x\,|\,x<\ct_{n-1}}+\Prob{x\geq\ct_{n-1}}(\Exp{x\,|\,x\geq\ct_{n-1}}+\E_{n-1}).
\]
Using 
\[
\Exp x=\Prob{x<\ct_{n-1}}\Exp{x\,|\,x<\ct_{n-1}}+\Prob{x\geq\ct_{n-1}}\Exp{x\,|\,x\geq\ct_{n-1}},
\]
this yields
\begin{eqnarray*}
\ct_{n} & = & \Exp x+\frac{1}{n}\Prob{x\geq\ct_{n-1}}\bigl(\E_{n-1}-(n-1)\Exp{x\,|\,x\geq\ct_{n-1}}\bigr)\\[6pt]
 & = & \Exp x+\frac{n-1}{n}\Prob{x\geq\ct_{n-1}}(\ct_{n-1}-\Exp{x\,|\,x\geq\ct_{n-1}}).
\end{eqnarray*}

Using Lemma~\ref{lem:sequential_G} (with $\ct_{n-1}\leq\ct_{n-2}$ by induction) and the fact that the second term is negative, we obtain
\begin{align*}
\tau_{n} & \leq  \Exp x+\frac{n-2}{n-1}\Prob{x\geq\ct_{n-1}}(\ct_{n-1}-\Exp{x\,|\,x\geq\ct_{n-1}})\\[6pt]
 & \leq  \Exp x+\frac{n-2}{n-1}\Prob{x\geq\ct_{n-2}}(\ct_{n-2}-\Exp{x\,|\,x\geq\ct_{n-2}})\\
 & = \ct_{n-1},
\end{align*}
which concludes the proof. \hfill \Halmos
\end{proof}

We derive the optimal competitive ratio for the case where $X=\U[0,1]$.

\begin{restatable}{lemma}{lemseqEsim}
For $X=\U[0,1]$, we have\label{lem:sequential_Esimplification}
\[
\E_{n}=\begin{cases}
\nicefrac{1}{2}, & \mathrm{for\,}n=1,\\
\E_{n-1}+\nicefrac{1}{2}-\frac{\E_{n-1}^{2}}{2(n-1)}, & \mathrm{for\,}n>1.
\end{cases}
\]
\end{restatable}

\begin{proof}{Proof.}
The case $n=1$ follows from $\Exp x=\nicefrac{1}{2}$. For $n>1$,
we use the fact that Algorithm~\ref{alg:sequential} is optimal.
We obtain
\begin{align*}
\E_{n} & = n\Prob{x<\ct_{n-1}}\Exp{x\,|\,x<\ct_{n-1}}+\Prob{x\geq\ct_{n-1}}(\Exp{x\,|\,x\geq\ct_{n-1}}+\E_{n-1})\\[6pt]
 & = n \ct_{n-1}\cdot\frac{1}{2}\ct_{n-1}+(1-\ct_{n-1})\left(\frac{1+\ct_{n-1}}{2}+\E_{n-1}\right)\\[6pt]
 & = \frac{n\ct_{n-1}^{2}}{2}+\frac{1}{2}+\E_{n-1}-\frac{1}{2}\ct_{n-1}^{2}-\E_{n-1}\ct_{n-1}\\[6pt]
 & = \E_{n-1}+\frac{1}{2}-\frac{\E_{n-1}^{2}}{2(n-1)},
\end{align*}
which concludes the proof. \hfill \Halmos
\end{proof}

With this, we can bound the expected cost of any online algorithm.

\begin{lemma}
\label{lem:EnInSqrt}
For $X = \mathcal{U}[0,1]$, we have $\sqrt{n+1}-1\leq\E_{n}\leq\sqrt{n}$.
\end{lemma}

\begin{proof}{Proof.}
For the sake of contradiction, assume that $\E_{n}>\sqrt{n}$ for
some value of $n$. With Lemma~\ref{lem:sequential_Esimplification},
we obtain
\[
\E_{n+1}=\E_{n}+\frac{1}{2}-\frac{\E_{n}^{2}}{2n}<\E_{n}+\frac{1}{2}-\frac{n}{2n}=\E_{n},
\]
which is a contradiction with $\E_{n}$ being non-decreasing.

Let $h(n):=\sqrt{n+1}-1$. It is easy to check that $\E_{n}\geq h(n)$
for $n<7$. For $n\geq7$, we use induction on $n$. To that end,
assume $\E_{n}\geq h(n)$ holds and consider $\E_{n+1}$. Clearly,
$\E_{n+1}\geq\E_{n}$. If $\E_{n}\geq\sqrt{n+1}-0.8$, it thus suffices
to show that $h(n+1)-h(n)\leq0.2$. Since $h$ is concave and $n\geq7$,
we indeed have 
\[
h(n+1)-h(n)\leq h'(n)=\frac{1}{2\sqrt{n+1}}\leq0.2.
\]
Finally, let $\E_{n}<\sqrt{n+1}-0.8$. Using $n\geq7$, we show that
$\E_{n}$ grows faster than $h(n)$:
\begin{align*}
\E_{n+1}-\E_{n} & = \frac{1}{2}-\frac{\E_{n}^{2}}{2n} \geq \frac{1}{2}-\frac{(\sqrt{n+1}-0.8)^{2}}{2n} = \frac{160\sqrt{n+1}-164}{200n}\\[6pt]
 & \geq \frac{\sqrt{n+1}}{2(n+1)} = h'(n) \geq h(n+1)-h(n),
\end{align*}
which concludes the proof. \hfill \Halmos
\end{proof}

Together with Lemma~\ref{lem:opt_uniform}, we immediately get the following bound on the competitive ratio of any online algorithm.

\begin{restatable}{theorem}{thmsequentialsqrtlog}\label{thm:sequential_CR}
The competitive ratio of the best online algorithm for sequential
employment and a uniform distribution $X=\U[0,1]$ is $\Theta\left(\sqrt{n} / \log n \right)$.
\end{restatable}

\section{Analysis of the Markov Chains}
\label{sec:markov}

In this section, we study the Markov chains that govern the evolution of the threshold values of our algorithms. 

\subsection{Markov Chain \boldmath$\hat{M}(p,k)$}
\label{sec:markov_a}

We start with the simple Markov chain $\hat{M}(p,k)$ used in \S~\ref{sec:uniform_first} and \S~\ref{sec:unknown_distribution}. The Markov chain has states $0,\dots,k$ and transition probabilities as shown in Figure~\ref{fig:markov_a_appendix}.

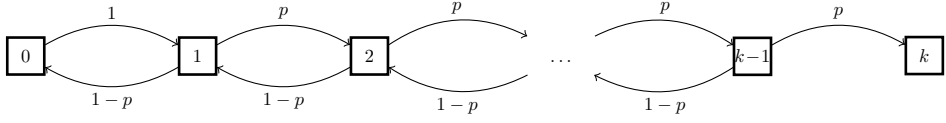
\begin{figure}[bt]
 
\centering
\begin{tikzpicture}[scale=0.3, every node/.style={scale=0.7}]
\matrix (m) [matrix of nodes, row sep=3em, column sep=5em,
	nodes={draw, rectangle, line width=1, text centered, minimum size=2em}]
    { \\ $0$ & $1$ & $2$  & |[draw=none]|$\phantom{1k}\dots\phantom{1k}$ & $\!\!k\!-\!1\!\!$  & $k$ \\ };

\draw[->] (m-2-1) edge [bend left] node [midway,above] {$1$} (m-2-2);
\draw[->] (m-2-2) edge [bend left] node [midway,below] {$1-p$} (m-2-1);

\draw[->] (m-2-2) edge [bend left] node [midway,above] {$p$} (m-2-3);
\draw[->] (m-2-3) edge [bend left] node [midway,below] {$1-p$} (m-2-2);

\draw[->] (m-2-3) edge [bend left] node [midway,above] {$p$} (m-2-4);
\draw[->] (m-2-4) edge [bend left] node [midway,below] {$1-p$} (m-2-3);

\draw[->] (m-2-4) edge [bend left] node [midway,above] {$p$} (m-2-5);
\draw[->] (m-2-5) edge [bend left] node [midway,below] {$1-p$} (m-2-4);

\draw[->] (m-2-5) edge [bend left] node [midway,above] {$p$} (m-2-6);


\foreach \x [count=\y] in {2,3}{
}

\end{tikzpicture}

\caption{Markov chain used in \S~\ref{sec:uniform_first} and \S~\ref{sec:unknown_distribution}. Nodes correspond to states. \label{fig:markov_a_appendix}}
\end{figure}

In the following, we compute the expected number of visits to each state.
\begin{lemma}
\label{lem:markov_a_visit}
Let $p > 1/2$ and $k \in \mathbb{N}$. Starting in state $0$, the expected number of visits to each state~$j$ of the Markov chain $\hat{M}(p,k)$ is at most $\frac{1}{2p-1}$.	
\end{lemma}


\begin{proof}{Proof.}
Let $v_{j}$ denote the
expected number of visits to state $j$, when starting from state
$0$. We derive that the values $v_j$, $j \in \{0,\dots,k\}$ satisfy the following equations
\begin{subequations} \label{eq:recurrence_exp_visits}
\begin{align}
v_{k} & = 1,\label{eq:recurrence_exp_visits_a1}\\[6pt]
v_{k} & = pv_{k-1}, \label{eq:recurrence_exp_visits_a2}\\[6pt]
v_{k-1} & = p v_{k-2},\label{eq:recurrence_exp_visits_b}\\[6pt]
v_{j} & =(1-p)v_{j+1}+p v_{j-1} & & \text{ for all }j\in\{2,\dots,k-2\},\label{eq:recurrence_exp_visits_c}\\[6pt]
v_{1} & = v_{0}+(1-p) v_{2},\label{eq:recurrence_exp_visits_d}\\[6pt]
v_{0} & =1 + (1-p)v_{1},\label{eq:recurrence_exp_visits_e}
\end{align}
\end{subequations} where \eqref{eq:recurrence_exp_visits_a1} follows
from the fact that $k$ is the absorbing state, \eqref{eq:recurrence_exp_visits_a2} uses that state $k$ is reached only from state $k-1$. Equation \eqref{eq:recurrence_exp_visits_b}
follows since state $k-1$ can be reached from state
$k-2$ only. Equation \eqref{eq:recurrence_exp_visits_c} follows from the fact, that we reach state $j$ from $j-1$ and $j+1$ and leave states $j-1$ and $j+1$ to $j$ with
a probability of $p$ and~$1-p$, respectively. As state~0 is left
with probability~$1$ towards its successor, Equation~(\ref{eq:recurrence_exp_visits_d})
holds as special case. Further, for state $0$, we get Equation \eqref{eq:recurrence_exp_visits_e}
since 0 is the starting state and can only be reached from state $1$.

Note that \eqref{eq:recurrence_exp_visits_a1} and \eqref{eq:recurrence_exp_visits_a2} imply $v_{k-1} = 1/p$ which by \eqref{eq:recurrence_exp_visits_b} implies $v_{k-2} = 1/p^2$. With these start values \eqref{eq:recurrence_exp_visits_c} uniquely defines a homogenous recurrence relation on $v_{1}, \dots, v_{k-1}$ with
\begin{align*}
v_{j} &= \frac{1}{p}v_{j+1} - \frac{1-p}{p}v_{j+2}	& &\text{ for all $j\in \{2,\dots,k-2\}$}.
\intertext{Solving this recurrence by the method of characteristic equations yields that the characteristic polynomial $x^2 - \frac{1}{p}x + \frac{1-p}{p}$ has roots $\frac{1}{p}-1$ and $1$ so that the explicit solution is}
v_{j} &= \lambda_1 \biggl(\frac{1}{p} -1\biggr)^{\!\!k-j-1} + \lambda_2	
\end{align*}
for some parameters $\lambda_1, \lambda_2 \in \mathbb{R}$. Choosing $\lambda_1$ and $\lambda_2$ such that the equations $v_{k-1} = 1/p$ and $v_{k-2} = 1/p^2$ are satisfied gives
\begin{align*}
\lambda_1 &= \frac{1}{2p - 1} \biggl(1- \frac{1}{p}\biggr), & \lambda_2 &= \frac{1}{p} - \frac{1}{2p - 1} \biggl(1- \frac{1}{p}\biggr).
\end{align*}
As a result, for~$j \in \{1,\dots,k\}$, we obtain
\begin{align}
v_{j} &=  \frac{1}{2p - 1} \biggl(1- \frac{1}{p}\biggr)\biggl(\frac{1}{p} -1\biggr)^{\!\!k-j-1} + \frac{1}{p} - \frac{1}{2p - 1} \biggl(1- \frac{1}{p}\biggr)\notag\\[6pt]
&=	\frac{1}{2p-1}\biggl[\biggl(\frac{1}{p}-1\biggr)-\biggl(\frac{1}{p}-1\biggr)^{\!\!k-j} \biggr] + \frac{1}{p}. \label{eq:recursion_gone}
\end{align}

Finally, $v_0$ is defined via \eqref{eq:recurrence_exp_visits_e}. Observe that, together with \eqref{eq:recursion_gone}, this satisfies \eqref{eq:recurrence_exp_visits_d} as required.

It remains to show that $v_j \leq \frac{1}{2p-1}$ for all $j \in \{0,\dots,k\}$. For $j\in \{1,\ldots,k-1\}$
we use Equation~(\ref{eq:recursion_gone}) and the fact that $p>\notnicefrac{1}{2}$ to obtain
\begin{align*}
v_{j} & =\frac{1}{2p-1}\left[\biggl(\frac{1}{p}-1\biggr) -\left(\frac{1}{p}-1\right)^{k-j}\right]+\frac{1}{p}\\[6pt]
 & \leq\frac{1}{2p-1}\left(\frac{1}{p}-1\right)+\frac{1}{p}\\
 & = \frac{p}{p\left(2p-1\right)} =\frac{1}{2p-1}.
\end{align*}
For $j=0$ we have by Equation~\eqref{eq:recurrence_exp_visits_e}
\begin{align*}
v_0 = 1 + (1-p)v_1 	\leq 1 +\frac{1-p}{2p-1} = \frac{p}{2p-1} \leq \frac{1}{2p-1}
\end{align*}
which completes the proof. \hfill \Halmos
\end{proof}


\subsection{Markov Chain \boldmath$\hat{N}(p,k)$}
\label{sec:markov_b}

In this section, we study the Markov chain $\hat{N}(p,k)$ used in \S~\ref{sec:uniform2} and \S~\ref{arbitrary_distributions}. The Markov chain has states $A_j$ and $B_j$, for $j\in\{0,\dots,k\}$ and transition probabilities as shown in Figure~\ref{fig:markov_b_appendix}. 

We start to bound the expected number of transitions from an $A$-state to a $B$-state.

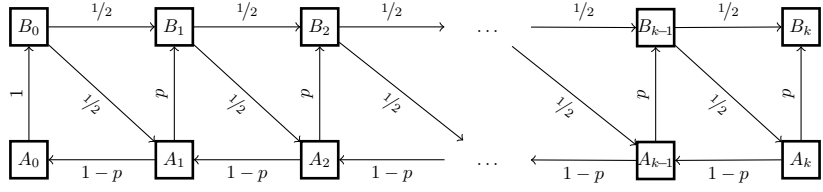
\begin{figure}[bt]
    \centering
\begin{tikzpicture}[scale=0.3, every node/.style={scale=0.7}]
\matrix (m) [matrix of nodes, row sep=3.5em, column sep=4em,
	nodes={draw, rectangle, line width=1, text centered, minimum size=2em}]
    { $B_0$ & $B_1$ & $B_2$ &|[draw=none]| $\phantom{B_1}\dots\phantom{B_1}$  & $\!\!\!B_{k\!-\!1}\!\!\!$ & $B_{k}$ \\
      $A_0$ & $A_1$ & $A_2$ &|[draw=none]| $\phantom{A_1}\dots\phantom{A_1}$  & $\!\!\!A_{k\!-\!1}\!\!\!$  & $A_{k}$ \\ };
\draw[->] (m-2-1) -- (m-1-1) node [above,midway,sloped] {$1$};
\draw[->] (m-2-2) -- (m-1-2) node [above,midway,sloped] {$p$};
\draw[->] (m-2-3) -- (m-1-3) node [above,midway,sloped] {$p$};
\draw[->] (m-2-5) -- (m-1-5) node [above,midway,sloped,rotate=180] {$p$};
\draw[->] (m-2-6) -- (m-1-6) node [above,midway,sloped,rotate=180] {$p$};

\draw[->] (m-2-2) -- (m-2-1) node [below,midway] {$1-p$};
\draw[->] (m-2-3) -- (m-2-2) node [below,midway] {$1-p$};
\draw[->] (m-2-4) -- (m-2-3) node [below,midway] {$1-p$};
\draw[->] (m-2-5) -- (m-2-4) node [below,midway] {$1-p$};
\draw[->] (m-2-6) -- (m-2-5) node [below,midway] {$1-p$};

\foreach \x [count=\y] in {2,3,4,5,6}{
    \draw[->] (m-1-\y) -- (m-1-\x) node [above,midway,sloped] {$\nicefrac{1}{2}$};
    \draw[->] (m-1-\y) -- (m-2-\x) node [below,midway,sloped] {$\nicefrac{1}{2}$};
}


\end{tikzpicture}
\caption{Markov chain $\hat{N}(p,k)$ with homogenous transition probability $p$ and $k+1$ states used in \S~\ref{sec:uniform2} and \S~\ref{arbitrary_distributions}. Nodes correspond to states. \label{fig:markov_b_appendix}}
\end{figure}


\begin{lemma}
\label{lem:markov_b_transitions1}
Starting in state $A_0$ of Markov chain $\hat{N}(p,k)$, the expected number of transitions from an $A$-state to a $B$-state is at most
\begin{align*}
h = \frac{kp}{3p-1} - \frac{4p(1-2p)}{(3p-1)^2} + \biggl(\frac{1-p}{3p-1}\biggr)^{\!\!2}\!\!\biggl(\frac{2(1-p)}{1+p}\biggr)^{\!\!k}.	
\end{align*}
\end{lemma}

\begin{proof}{Proof.}
Let $a_{j}$ (respectively~$b_{j}$) denote the expected
number of transitions from an $A$-state to a $B$-state, when starting
from state $A_{j}$ (respectively~$B_{j}$). We get \begin{subequations}
\label{eq:recurrence2} 
\begin{align}
b_{k} & =0,\label{eq:recurrenceB1}\\[6pt]
b_{j} & =\frac{1}{2}b_{j+1}+\frac{1}{2}a_{j+1} &  & \text{ for all }j\in\{0,\dots,k-1\},\label{eq:recurrenceB2}\\[6pt]
a_{j} & =p(b_{j}+1)+(1-p) a_{j-1} &  & \text{ for all }j\in\{1,\dots,k\},\label{eq:recurrenceA1}\\[6pt]
a_{0} & =1+b_{0}.\label{eq:recurrenceA2}
\end{align}
\end{subequations}
Defining $\beta=\frac{2(1-p)}{1+p}$, for $j\in\{0,\ldots,k\}$, it is straightforward to check that (\ref{eq:recurrenceB1}), (\ref{eq:recurrenceB2}) and
(\ref{eq:recurrenceA1}) are fulfilled by 
\begin{align*}
a_{j} & =\frac{(k-j+2)p}{3p-1}-\beta^{j}\frac{2p(1-p)}{(3p-1)^{2}}+\beta^{k}\frac{(1-p)^{2}}{(3p-1)^{2}},\quad\text{and}\\[6pt]
b_{j} & =\frac{(k-j)p}{3p-1}-\beta^{j}\frac{(1-p)^{2}}{(3p-1)^{2}}+\beta^{k}\frac{(1-p)^{2}}{(3p-1)^{2}}.
\end{align*}
It follows that the expected number of transitions from an $A$-state
to a $B$-state when starting at $A_{0}$ is 
\begin{align*}
a_{0} & =\frac{(k+2)p}{3p-1}-\frac{2p(1-p)}{(3p-1)^{2}}+\beta^{k}\frac{(1-p)^{2}}{(3p-1)^{2}} \\
&=\frac{kp}{3p-1}-\frac{4p(1-2p)}{(3p-1)^{2}}+\frac{2^{k}(1-p)^{k+2}}{(3p-1)^{2}(1+p)^{k}},
\end{align*}
which completes the proof. \hfill \Halmos
\end{proof}


\begin{lemma}
\label{lem:markov_b_transitions2}
Starting in state $A_0$ of Markov chain $\hat{N}(p,k)$ for each $j \in \{0,\dots,k\!-\!1\}$ the expected number of transitions from $B_j$ to $A_{j+1}$ is at most $\frac{p}{3p-1}$.
\end{lemma}

\begin{proof}{Proof.}
As the expected number of such transitions is half the expected number of 
visits to state~$B_j$, it suffices to bound the latter quantity.

Suppose we are in state $B_j$. 
The probability of coming back to $B_j$ equals the probability of hitting~$A_j$ from~$B_j$.
Denote by $a_i(j)$, $b_i(j)$ the hitting probability of state~$A_i$ from~$A_j$ and $B_j$, respectively. We have
\begin{subequations}
\label{eq:recurrence_arbitrary} 
\begin{align}
b_i(k) & =0,\label{eq:recurrence_arbitrary1}\\[6pt]
b_i(j) & =\frac{1}{2} b_i(j+1)+\frac{1}{2}a_i(j+1) &  & \text{ for all }j\in\{i,\dots,k-1\},\label{eq:recurrence_arbitrary2}\\[6pt]
a_i(j) & =p b_i(j)+(1-p) a_i(j-1) &  & \text{ for all }j\in\{i+1,\dots,k\},\label{eq:recurrence_arbitrary3}\\[6pt]
a_i(i) & =1\label{eq:recurrence_arbitrary4}
\end{align}
\end{subequations}
Let $\beta= \frac{2(1-p)}{p+1} < 1$ (as $p>\nicefrac{1}{3}$). It is easy to check that for $j\in\{i, \ldots, k\}$
\begin{align*}
a_i(j) & \le \beta^{j-i} \\
b_i(j) & \le \frac{1-p}{2p}\beta^{j-i} ,
\end{align*}
gives an upper bound on the solution of \eqref{eq:recurrence_arbitrary}
as these values satisfy equalities \eqref{eq:recurrence_arbitrary2}, \eqref{eq:recurrence_arbitrary3}, \eqref{eq:recurrence_arbitrary4}, and only overestimate \eqref{eq:recurrence_arbitrary1}.
We can interpret the visits to state~$B_j$ after the first visit as a geometric random variable with success probability~$1-b_j(j)$.
Thus, the expected number of visits to~$B_j$ is given by
\begin{align*}
1 + \frac{1-(1-b_j(j))}{1-b_j(j)} = \frac{1}{1-b_j(j)} & \leq \frac{1}{1-\frac{1-p}{2p}} = \frac{2p}{3p-1}. 
\end{align*}
We conclude that the expected number of transitions from $B_j$ to $A_{j+1}$ is at most $\frac{p}{3p-1}$, proving the claim. \hfill \Halmos
\end{proof}

\section{Conclusion}
\label{sec:conclusion}

We considered prophet inequalities with a covering constraint and a minimization objective. We gave constant competitive algorithms for this type of problem and established concurrent employment as a necessary feature of such algorithms.

We note that our results extend to slightly more general settings, where 
(a) we relax the covering constraint by associating a penalty~$B < \infty$ with time steps where no contract is active,
(b) multiple applicants arrive in each time step,
(c) applicants may be hired fractionally.

A crucial limitation of our model is the assumption that costs are distributed independently, and it remains an interesting question how to address correlated costs.

\bibliographystyle{informs2014} 
\bibliography{secretary_leasing_literature} 

\begin{thebibliography}{25}
\providecommand{\natexlab}[1]{#1}
\providecommand{\url}[1]{\texttt{#1}}
\providecommand{\urlprefix}{URL }

\bibitem[{Alaei(2014)}]{alaei2014}
Alaei S (2014) Bayesian combinatorial auctions: Expanding single buyer
  mechanisms to many buyers. \emph{SIAM. J. Comput.} 43(2):930--972.

\bibitem[{Bearden(2006)}]{bearden}
Bearden JN (2006) A new secretary problem with rank-based selection and
  cardinal payoffs. \emph{J. Math. Psychol.} 50:58--59.

\bibitem[{Bellman(1954)}]{bellman1954}
Bellman R (1954) The theory of dynamic programming. \emph{Bull. Amer. Math.
  Soc.} 60(6):503--514.

\bibitem[{Cayley(1875)}]{cayley1875}
Cayley AF (1875) \emph{Mathematical questions with their solutions}, volume~10,
  587--588 (Cambridge University Press).

\bibitem[{Chawla et~al.(2010)Chawla, Hartline, Malec, , \protect\BIBand{}
  Sivan}]{chawla2010}
Chawla S, Hartline JD, Malec DL, , Sivan B (2010) Multi-parameter mechanism
  design and sequential posted pricing. \emph{Proceedings of the 41th Annual
  ACM Symposium on Theory of Computing (STOC)}, 311--320.

\bibitem[{D\"utting \protect\BIBand{} Kleinberg(2015)}]{duetting2015}
D\"utting P, Kleinberg R (2015) Polymatroid prophet inequalities. Bansal N,
  Finocchi I, eds., \emph{Proceedings of the 23rd European Symposium on
  Algorithms (ESA)}, volume 9294 of \emph{Lecture Notes in Computer Science},
  437--449.

\bibitem[{Esfandiari et~al.(2015)Esfandiari, Hajiaghayi, Liaghat,
  \protect\BIBand{} Monemizadeh}]{esfandiari2015}
Esfandiari H, Hajiaghayi M, Liaghat V, Monemizadeh M (2015) Prophet secretary.
  Bansal N, Finocchi I, eds., \emph{Proceedings of the 23rd European Symposium
  on Algorithms (ESA)}, volume 9294 of \emph{Lecture Notes in Computer
  Science}, 496--508.

\bibitem[{Ferguson(1989)}]{ferguson1989}
Ferguson TS (1989) Who solved the secretary problem? \emph{Statistical Science}
  4(3):282--289.

\bibitem[{Fiat et~al.(2015)Fiat, Gorelik, Kaplan, \protect\BIBand{}
  Novgorodov}]{fiat2015}
Fiat A, Gorelik I, Kaplan H, Novgorodov S (2015) The temp secretary problem.
  \emph{Proceeding of the 23rd European Symposium on Algorithms (ESA)},
  631--642.

\bibitem[{Gilbert \protect\BIBand{} Mosteller(1966)}]{GilbertMosteller/66}
Gilbert JP, Mosteller F (1966) Recognizing the maximum of a sequence.
  \emph{Journal of the American Statistical Association} 61(313):35--73.

\bibitem[{G\"obel et~al.(2014)G\"obel, Hoefer, Kesselheim, Schleiden,
  \protect\BIBand{} V\"ocking}]{goebel2014}
G\"obel O, Hoefer M, Kesselheim T, Schleiden T, V\"ocking B (2014) Online
  independent set beyond the worst-case: Secretaries, prophets, and periods.
  \emph{Proceedings of the 41st International Colloquium on Automata,
  Languages, and Programming (ICALP)}, 508--519.

\bibitem[{Guttman(1960)}]{guttman1960}
Guttman I (1960) On a problem of {L. Moser}. \emph{Canad. Math. Bull.}
  3:35--39.

\bibitem[{Hajiaghayi et~al.(2007)Hajiaghayi, Kleinberg, \protect\BIBand{}
  Sandholm}]{hajiaghayi2007}
Hajiaghayi MT, Kleinberg R, Sandholm T (2007) Automated online mechanism design
  and prophet inequalities. \emph{Proceedings of the 22nd national conference
  on Artificial intelligence (AAAI)}, 58--65.

\bibitem[{Hill \protect\BIBand{} Kertz(1992)}]{hill1992}
Hill TP, Kertz RP (1992) A survey of prophet inequalities in optimal stopping
  theory. \emph{Contemporary Mathematics} 125:191--207.

\bibitem[{Karlin(1962)}]{karlin1962}
Karlin S (1962) Stochastic models and optimal policy for selling an asset.
  Arrow KJ, Karlin S, Scarf H, eds., \emph{Studies in Applied Probability and
  Management Science}, 148--158 (Stanford University Press).

\bibitem[{Kennedy(1987)}]{kennedy1987}
Kennedy D (1987) Prophet-type inequalities for multi-choice optimal stopping.
  \emph{Stochastic Processes and their applications} 24:77--88.

\bibitem[{Kleinberg \protect\BIBand{} Weinberg(2012)}]{kleinberg2012}
Kleinberg R, Weinberg SM (2012) Matroid prophet inequalities. \emph{Proceedings
  of the 44th Annual ACM Symposium on Theory of Computing (STOC)}, 123--136.

\bibitem[{Krengel \protect\BIBand{} Sucheston(1977)}]{krengel77}
Krengel U, Sucheston L (1977) Semiamarts and finite values. \emph{Bull. Amer.
  Math. Soc.} 83:745--747.

\bibitem[{Krengel \protect\BIBand{} Sucheston(1978)}]{krengel78}
Krengel U, Sucheston L (1978) On semiamarts, amarts, and processes with finite
  value. \emph{Advances in Prob.} 4:197--266.

\bibitem[{Moser(1956)}]{moser1956}
Moser L (1956) On a problem of {Cayley}. \emph{Scripta Mathematica}
  22:289--292.

\bibitem[{Myerson(1981)}]{myerson1981}
Myerson RB (1981) Optimal auction design. \emph{Math. Oper. Res.} 6:58--73.

\bibitem[{Rubinstein(2016)}]{rubinstein2016}
Rubinstein A (2016) Beyond matroids: Secretary problem and prophet inequality
  with general constraints. \emph{Proceedings of the 48th Annual ACM Symposium
  on Theory of Computing (STOC)}, 324--332.

\bibitem[{Rubinstein \protect\BIBand{} Singla(2017)}]{rubinstein2017}
Rubinstein A, Singla S (2017) Combinatorial prophet inequalities.
  \emph{Proceedings of the 28th Annual ACM-SIAM Symposium on Discrete
  Algorithms (SODA)}, 1671--1687.

\bibitem[{Samuel-Cahn(1984)}]{samuel-cahn1984}
Samuel-Cahn E (1984) Comparison of threshold stop rules and maximum for
  independent non-negative random variables. \emph{Ann. Probab.} 12:1213--1216.

\bibitem[{Verroios et~al.(2015)Verroios, Papadimitriou, Johari,
  \protect\BIBand{} Garcia-Molina}]{verroios2015}
Verroios V, Papadimitriou P, Johari R, Garcia-Molina H (2015) Client clustering
  for hiring modeling in work marketplaces. \emph{Proceedings of the 21th ACM
  SIGKDD International Conference on Knowledge Discovery and Data Mining},
  2187--2196.

\end{thebibliography}


\end{document}